\documentclass[11pt]{article}
\usepackage[T1]{fontenc}
\usepackage[utf8]{inputenc}
\usepackage{lmodern}
\usepackage{xspace}                                     
\usepackage{amsfonts,amsmath,amssymb, amsthm, mathtools}
\usepackage{thm-restate}
\usepackage{dsfont} 
\usepackage{algorithmicx,algpseudocode,algorithm}
\usepackage[usenames,dvipsnames,table]{xcolor}
\usepackage{csquotes}
\usepackage{mfirstuc}

\usepackage{tikz}
\usetikzlibrary{arrows}
\usetikzlibrary{calc,decorations.pathmorphing,patterns}
\makeatletter
\let\pgfmathrandomX=\pgfmathrandom@
\usepackage{pgfplots}%
\let\pgfmathrandom@=\pgfmathrandomX
\makeatother
\RequirePackage{tikz, pgflibraryplotmarks}

\usepackage[backref,colorlinks,citecolor=blue,bookmarks=true,linktocpage]{hyperref}
\usepackage{aliascnt}
\usepackage[numbered]{bookmark}
\usepackage[capitalise]{cleveref}
\usepackage{fullpage}

\usepackage[shortlabels]{enumitem}
  \setitemize{noitemsep,topsep=3pt,parsep=2pt,partopsep=2pt} 
  \setenumerate{itemsep=1pt,topsep=2pt,parsep=2pt,partopsep=2pt}
  \setdescription{itemsep=1pt}
  
\usepackage{mleftright} 

\usepackage{tikz,forest}
\usetikzlibrary{arrows.meta}
\PassOptionsToPackage{hyphens}{url}\usepackage{hyperref}
 \makeatletter
  \theoremstyle{plain} 
  	\newtheorem{theorem}{Theorem}[section]
  	\newtheorem*{theorem*}{Theorem} 
  	\newaliascnt{coro}{theorem}
  	  \newtheorem{corollary}[coro]{Corollary}
  	\aliascntresetthe{coro}
  	\newaliascnt{lem}{theorem}
  		\newtheorem{lemma}[lem]{Lemma}
  	\aliascntresetthe{lem}
  	\newaliascnt{clm}{theorem}
  		\newtheorem{claim}[clm]{Claim}
	\aliascntresetthe{clm}
	\newaliascnt{fact}{theorem}
 	 	
	\aliascntresetthe{fact}
  	
  \newaliascnt{prop}{theorem}
  		\newtheorem{proposition}[prop]{Proposition}
	\aliascntresetthe{prop}
	\newaliascnt{conj}{theorem}
  		
	\aliascntresetthe{conj}
 	 
  	\newtheorem{question}[theorem]{Question}
  	\newtheorem*{question*}{Question}
  \theoremstyle{remark} 
  	\newtheorem{remark}[theorem]{Remark}
  	
 	 \newtheorem{example}[theorem]{Example}
  \theoremstyle{definition} 
  	\newaliascnt{defn}{theorem}
 		 \newtheorem{definition}[defn]{Definition}
 	 \aliascntresetthe{defn}
 	 \newtheorem{observation}[theorem]{Observation}

\newenvironment{proofof}[1]{\begin{proof}[Proof of {#1}]}{\end{proof}}

\crefname{claim}{Claim}{Claims}

\providecommand{\email}[1]{\href{mailto:#1}{\nolinkurl{#1}\xspace}}

\newcommand{\eps}{\ensuremath{\varepsilon}\xspace}

\newcommand{\property}{\ensuremath{\mathcal{P}}\xspace} 
\newcommand{\task}{\ensuremath{\mathcal{P}}\xspace} 
\newcommand{\class}{\ensuremath{\mathcal{C}}\xspace} 
\newcommand{\eqdef}{:=}

\newcommand{\accept}{\textsf{accept}\xspace}

\newcommand{\reject}{\textsf{reject}\xspace}

\newcommand{\domain}{\ensuremath{\mathcal{X}}\xspace} 
\newcommand{\distribs}[1]{\Delta\!\left(#1\right)} 

\newcommand{\no}{{\sf{}no}\xspace}

\newcommand{\Testing}{\ensuremath{\mathcal{T}}\xspace} 
\newcommand{\UTesting}{\ensuremath{\mathcal{T}^{\rm u}}\xspace} 
\newcommand{\Learning}{\ensuremath{\mathcal{L}}\xspace}

\newcommand{\bigO}[1]{{O\mleft( #1 \mright)}}

\newcommand{\bigTheta}[1]{{\Theta\mleft( #1 \mright)}}
\newcommand{\bigOmega}[1]{{\Omega\mleft( #1 \mright)}}

\newcommand{\tildeO}[1]{\tilde{O}\mleft( #1 \mright)}

\newcommand{\setOfSuchThat}[2]{ \left\{\; #1 \;\colon\; #2\; \right\} } 			
\newcommand{\indicSet}[1]{\mathds{1}_{#1}}                                              
\newcommand{\indic}[1]{\indicSet{\left\{#1\right\}}}

\newcommand{\dtv}{\operatorname{d}_{\rm TV}}

\newcommand{\totalvardistrestr}[3][]{{\dtv^{#1}\!\left({#2, #3}\right)}}
\newcommand{\totalvardist}[2]{\totalvardistrestr[]{#1}{#2}}

\newcommand{\dist}[2]{\operatorname{dist}\mleft({#1, #2}\mright)}
\newcommand{\proba}{\Pr}
\newcommand{\probaOf}[1]{\proba\!\left[\, #1\, \right]}
\newcommand{\probaCond}[2]{\proba\!\left[\, #1 \;\middle\vert\; #2\, \right]}
\newcommand{\probaDistrOf}[2]{\proba_{#1}\left[\, #2\, \right]}

\newcommand{\expect}[1]{\mathbb{E}\!\left[#1\right]}
\newcommand{\expectCond}[2]{\mathbb{E}\!\left[\, #1 \;\middle\vert\; #2\, \right]}
\newcommand{\shortexpect}{\mathbb{E}}
\newcommand{\var}{\operatorname{Var}}
\newcommand{\bE}[2]{\shortexpect_{#1}{\left[#2\right]}}

\newcommand{\uniform}{\ensuremath{\mathbf{u}}}
\newcommand{\uniformOn}[1]{\ensuremath{\uniform_{ #1 } }}
\newcommand{\bernoulli}[1]{\ensuremath{\operatorname{Bern}( #1 ) }}

\newcommand{\poisson}[1]{\ensuremath{\operatorname{Poisson}\!\left( #1 \right) }}

\newcommand{\mutualinfo}[2]{ I\left(#1; #2\right) }

\newcommand{\norm}[1]{\lVert#1{\rVert}}
\newcommand{\normone}[1]{{\norm{#1}}_1}
\newcommand{\normtwo}[1]{{\norm{#1}}_2}
\newcommand{\norminf}[1]{{\norm{#1}}_\infty}
\newcommand{\abs}[1]{\left\lvert #1 \right\rvert}

\newcommand{\clg}[1]{\left\lceil #1 \right\rceil}

\newcommand{\R}{\ensuremath{\mathbb{R}}\xspace}

\newcommand{\N}{\ensuremath{\mathbb{N}}\xspace}
\newcommand{\lp}[1][1]{\ell_{#1}}

\newcommand{\p}{\mathbf{p}}
\newcommand{\q}{\mathbf{q}}

\newcommand{\cE}{\mathcal{E}} \newcommand{\cI}{\mathcal{I}}
\newcommand{\cP}{\mathcal{P}} \newcommand{\cS}{\mathcal{S}}
\newcommand{\cX}{\mathcal{X}} 
 
\usepackage{utopia} 

\title{Distributed Simulation and Distributed Inference}

\author{ Jayadev Acharya\thanks{Cornell University. Email: \email{acharya@cornell.edu}.} 
  \and Cl\'{e}ment L. Canonne\thanks{Stanford University. Email: \email{ccanonne@cs.stanford.edu}. Supported by a Motwani Postdoctoral Fellowship.} 
  \and Himanshu Tyagi\thanks{Indian Institute of Science. Email: \email{htyagi@iisc.ac.in}.}  }

\newcommand{\referee}{\mathcal{R}\xspace}
\newcommand{\ns}{n} 
\newcommand{\ab}{k} 
\newcommand{\numbits}{\ell}

\begin{document}

\maketitle

\begin{abstract}
  Independent samples from an unknown probability distribution $\p$ on a
domain of size $\ab$ are distributed across $\ns$ players, with each
player holding one sample. Each player can communicate $\numbits$ bits
to a central referee in a simultaneous message passing model of
communication to help the referee infer a property of the unknown 
$\p$. What is the least number of players for inference required in the communication-starved setting of $\numbits <\log \ab$? We begin by exploring a general \emph{simulate-and-infer} strategy for such inference problems where the center simulates the desired number of samples from the unknown distribution and applies standard inference algorithms for the collocated setting. Our first result shows that for $\numbits<\log \ab$ perfect simulation of even a single sample is not possible. Nonetheless, we present next a Las Vegas algorithm that simulates a single sample from the unknown  distribution using no more than $O(\ab/2^\numbits)$ samples in expectation. As an immediate corollary, it follows that simulate-and-infer attains the optimal sample complexity of $\Theta(\ab^2/2^{\numbits}\eps^2)$ for learning the unknown distribution to an accuracy of $\eps$ in total variation distance. 

For the prototypical testing problem of identity testing,
simulate-and-infer works with $O(\ab^{3/2}/2^{\numbits}\eps^2)$
samples, a requirement that seems to be inherent for all communication
protocols not using any additional resources. Interestingly, we can
break this barrier using public coins. Specifically, we exhibit a
public-coin communication protocol that accomplishes identity testing
using $O(\ab/\sqrt{2^{\numbits}}\eps^2)$ samples. Furthermore, we show
that this is optimal up to  constant factors. Our theoretically sample-optimal protocol is easy to implement in practice. Our proof of lower
bound entails showing a contraction in $\chi^2$ distance of product
distributions due to communication constraints and maybe of interest
beyond the current setting.
 \end{abstract}

\thispagestyle{empty}
\newpage
\setcounter{page}{1}

\tableofcontents

\newpage

\section{Introduction}

A set of sensor nodes are deployed in an active volcano to measure
seismic activity. They are connected to a central server over a low
bandwidth communication link, but owing to their very limited battery,
they can only send a fixed number of short packets. The server seeks
to determine if the distribution of the quantized measurements have
changed significantly from the one on record. How many sensors must be
deployed?

This situation is typical in many emerging sensor network applications
as well as other distributed learning scenarios where the data is
distributed across multiple clients with limited communication
capability. The question above is an instance of the \emph{distributed
inference} problem where independent samples from an unknown
distribution are given to physically separated players who can only
send a limited amount of communication to a central referee. The
referee uses the communication to infer some properties of the
generating distribution of the samples. A variant of this problem
where each player gets different (correlated) coordinates of the
independent samples has been studied extensively in the information
theory literature (cf.~\cite{AhlCsi86, Han87, HanAmari98}). The 
problem described above has itself received significant attention lately in various communities (see, for
instance,~\cite{BPCPE:11,BBFM:12,Shamir:14,DGLNOS:17,HOW:18}), with
the objective of performing parameter or density estimation while
minimizing the number of players (or, equivalently, the amount of
data). Of particular interest to us are the results
in \cite{DGLNOS:17,HOW:18}, which consider distribution learning problems. Specifically, it is shown that roughly trivial schemes quantizing and compressing each sample separately turn out to be
optimal for \emph{simultaneous message passing} (SMP) communication protocols. One of our goals in this research is to formalize 
this heuristic and {explore its limitations.

To formalize the question, we introduce a natural notion of
\emph{distributed simulation}: $\ns$ players observing an independent
sample each from an unknown $\ab$-ary distribution $\p$ can send
$\numbits$-bits each to a referee. A distributed simulation protocol consists
of an SMP and a randomized decision map that enables the referee to generate a sample from $\p$ using
the communication from the players. Clearly, when\footnote{We assume
throughout that $\log \ab$ is an integer.} $\numbits\geq \log \ab$ such a
sample can be obtained by getting the sample of any one player. But
what can be done in the communication-starved regime of 
$\numbits< \log \ab$?

This problem of distributed simulation is connected innately to the
aforementioned distributed inference problems where the generating distribution $\p$ is unknown and the
referee uses the communication from the players to accomplish a
specific inference task $\task$.
\begin{question} What is the minimum number of players $\ns$
required by an SMP that successfully accomplishes the inference $\task$, as a function of $\ab$, $\numbits$, and the relevant parameters of $\task$?
\end{question}
The formulation above encompasses both density and
parameter estimation, as well as {distribution testing} (see
e.g.~\cite{Rubinfeld:12:Survey,Canonne:15:Survey} and~\cite{Goldreich:17}
for a survey of property and distribution testing). 

Equipped with a distributed simulation, we can accomplish any distributed
inference task by simulating as many samples as warranted by the sample
complexity of the inference problem. Our objective is to
understand when such a \emph{simulate-and-infer} strategy is optimal. We study the distributed simulation problem and apply it
 to distribution learning and distribution testing. Our results are the most striking for distribution testing, which, to the best of our knowledge, has not been studied in the distributed setting prior to our work.

Starting with the distributed simulation problem, we first establish that perfect simulation is impossible using any finite number of players in the
communication-starved regime. This establishes an interesting
dichotomy where communication from a single party suffices for perfect
simulation when $\numbits
\geq \log \ab$ and no finite number of parties can accomplish it when
$\numbits<\log\ab$. If we allow a small probability of declaring failure,
namely Las Vegas schemes, distributed simulation is possible with
finitely many players. Indeed, we present such a Las Vegas distributed
simulation scheme that requires optimal (up to constant factors)
number of players to simulate $\ab$-ary distributions using
$\numbits$ bits of communication per player. Moving to the connection
between distributed simulation and distributed inference, we exhibit
an instance when simulate-and-infer {is} optimal. Perhaps more
interestingly, we even exhibit a case where a simple
simulate-and-infer scheme is far from optimal if we allow the
communication protocol to use public randomness. As a byproduct, we
 characterize the least number of players required for
distributed uniformity testing in the SMP model. We provide
a concrete description of our results in the next section, followed by
an overview of our proof techniques in the subsequent section. To put
our results in context, we provide an overview of the literature on
distribution learning as well.

\subsection{Main results}
Our first theorem shows that perfect distributed simulation with a
finite number of players is impossible:
\begin{theorem}\label{theo:distributed:simulation:impossiblity}
  For every $\ab \geq 1$ and $\numbits< \log\ab$, there exists no SMP
  with $\numbits$ bits of communication per player for distributed
  simulation over $[\ab]$ with finite number of players. Furthermore,
  the result continues to hold even when public-coin and interactive
  communication protocols are allowed.
\end{theorem}
In light of this impossibility result, one can ask if distributed
estimation is still possible by relaxing the requirement of finiteness
in the worst-case for the number of players. We demonstrate that this
is indeed the case and describe a protocol with finite {expected}
number of players.\footnote{Or, roughly equivalently, when one is
allowed to abort with a special symbol with small constant
probability.}
\begin{theorem}\label{theo:distributed:simulation}
  For every $\ab, \numbits\geq 1$, there exists a private-coin
  protocol with $\numbits$ bits of communication per player for
  distributed simulation over $[\ab]$ and expected number of players
  $\bigO{{\ab}/{2^\numbits}\vee 1}$. Moreover, this expected number is
  optimal, up to constant factors, even when public-coin and
  interactive communication protocols are allowed.
\end{theorem}
\noindent We use this distributed simulation result to
derive protocols for \emph{any} distributed inference task:
\begin{theorem}[Informal]\label{theo:inference:simulation:informal}
  For any inference task $\task$ over $\ab$-ary distributions with
sample complexity $s$ in the non-distributed model, there exists a
private-coin protocol for $\task$, with $\numbits$ bits of
communication per player, and $\ns=O({s\cdot \ab/2^\numbits})$
players.
\end{theorem}
\noindent Instantiating this general statement for the prototypical distribution testing problem of uniformity testing leads to: 
\begin{corollary}\label{theo:uniformity:private:randomness}
For every $\ab,\numbits\geq 1$, there exists a private-coin protocol
for testing uniformity over $[\ab]$, with $\numbits$ bits of
communication per player and
$\ns=\bigO{\frac{\ab^{3/2}}{(2^\numbits\wedge \ab)\eps^2}}$ players.
\end{corollary}
The optimality of the {simulate-and-infer} strategy that generates
$O(\sqrt{\ab})$ samples from the unknown $\p$ at the referee using
private-coin protocols is open. Note that for a general inference
problem even for $\ab$-ary observations the effective support-size can
be much smaller. Thus, we can define the size of a problem as the
least number of bits to which each samples can be compressed without
increase in the number of compressed sample required to solve the
problem (see~\cref{sec:conjecture} for a formal definition).  An
intriguing question ensues:
\begin{question}[The Flying Pony Question (Informal)]\label{conj:flying:poney}
Does the compressed simulate-and-infer scheme, which simulates
independent samples compressed to the size of the problem using
private-coin protocols and sends them to the referee who then infers
from them, require the least number of players?
\end{question}
For the problems considered in~\cite{DGLNOS:17,HOW:18}, the answer to
 the question above is in the affirmative. However, we exhibit an
 example in~\cref{sec:conjecture} for which the answer is
 negative. Roughly, the problem we consider is that of testing if the
 distribution is uniform on $[\ab]$ or instead satisfies the
 following: for every $i\in[\ab]$, $\p_{i}$ is either $0$ or
 $2/\ab$. We show that the size of this problem remains $\log \ab$,
 whereby the simple simulate-and-infer scheme of the question above
 for $\numbits=1$ will require $O(\ab^{3/2})$ players. On the other
 hand, one can obtain a simple scheme to solve this task using $1$-bit
 communication from only $O(\ab)$ players. Interestingly, even this
 new scheme is of simulate-and-infer form, although it compresses
 below the size of the problem. \smallskip

While the answer to the question above remains open for uniformity
testing using private-coin protocols, it is natural to examine its
scope and consider public-coin protocols for uniformity testing. As it
turns out, here, too, the answer to the question is
negative~--~public-coin protocols lead to an improvement in the
required number of parties over the simple simulate-and-infer protocol
described earlier by a factor of
$\sqrt{\ab/2^{\numbits}}$. Specifically, we provide a public-coin
protocol for uniformity testing that requires roughly
$O(\ab/2^{\numbits/2})$ players and show that no public-coin protocol
can work with fewer players.
\begin{theorem}\label{theo:uniformity:shared:randomness}
For every $\ab,\numbits\geq 1$, consider the problem of testing if the
distribution is uniform or $\eps$-far from uniform in total variation
distance.  There exists a public-coin protocol for uniformity testing
with $\numbits$ bits of communication per player and
$\ns=O\Big({\frac{\ab}{(2^{\numbits/2}\wedge \sqrt{\ab})\eps^2}}\Big)$
players. Moreover, this number is optimal up to constant factors.
\end{theorem}
\noindent In fact, we provide two different protocols achieving
this optimal guarantee. The first is remarkably simple to describe and
requires $\Omega(\numbits\cdot \ab)$ bits of shared randomness; the
second is more randomness-efficient, requiring only $O(2^\numbits\cdot \log \ab)$ bits of shared randomness,\footnote{For our regime of interest, $\numbits\ll \log \ab$, and so, $2^\numbits\cdot \log \ab \ll \numbits \cdot\ab$.} but it is also more involved.\smallskip

Before concluding this section, we emphasize that all our results for uniformity testing immediately imply the analogue for the more general  question of identity testing, via a standard reduction argument. We detail this further in~\cref{sec:uniformity}.

\subsection{Proof techniques}
We now provide a high-level description of the proofs of our main
results.
\paragraph{Perfect and $\alpha$-simulation.} Our
general impossibility result for perfect simulation with a finite
number of players is based on simple heuristics. Observe that for any
distribution $\p$ with $\p_i=0$ for some $i$, the referee must not
output $i$ for any sequence of received messages from the
players. However, since $\numbits<\log \ab$, by the pigeonhole
principle one can find a sequence of messages $M=(M_1, ..., M_{\ns})$
where each message $M_i$ has a positive probability of appearing from
two different elements in $[\ab]$. Note that there exist distributions
for which $M$ can occur with a positive probability, and not being
able to abort with the symbol $\bot$, upon receiving this sequence the
referee must output \emph{some} element, say $i^\ast$. Then, for any
distribution with $\p_{i^\ast}=0$, this sequence $M$ must not be
sent. But by construction each message in $M$ can be triggered by at
least two elements in $[\ab]$. Thus, we can find a distribution with
$\p_{i^\ast}=0$ for which the sequence of messages $M$ will be sent
with positive probability, which is a contradiction.

Next, we consider $\alpha$-simulation protocols, namely simulation
protocols that are allowed to abort with probability less than
$\alpha$. The proof of the positive result establishing the existence
of $\alpha$-simulation proceeds by dividing the alphabet into
$\ab/(2^\numbits-1)$ sets of size $2^\numbits-1$ and assigning each
such set to two different players (each using their $\numbits$ bits to
indicate whether their sample fell in this subset, and if so on which
element). If only one pair of players finds the sample in its assigned
subset, the referee can declare this as the output, and it will have
the desired probability. But it is possible that several pairs of
players observe their assigned symbol and send conflicting
messages. In this case, the referee cannot decide which of the
elements to choose and must declare abort.  However, we show that this
happens with a probability that depends only on the $\lp[2]$ norm of
the unknown distribution $\p$; if we could assume this norm to be
bounded away from $1$, then our protocol would require $O(\ab)$
players. Unfortunately, this need not be the case. To circumvent this
difficulty, we artificially duplicate every element of the domain and
``split'' each element $i\in[\ab]$ into two equiprobable elements
$i_1,i_2\in[2\ab]$. This has the effect of decreasing the $\lp[2]$
norm of $\p$ by a factor $\sqrt{2}$, allowing us to instead apply our
protocol to the resulting distribution $\p'$ on $[2\ab]$, for which
the aforementioned probability of aborting can be bounded by a
constant.

\paragraph{Distributed uniformity testing.} 
To test 
whether an unknown distribution $\p$ is uniform using at most $\numbits$
bits to describe each sample, a natural idea is to randomly
partition the alphabet into $L\eqdef 2^\numbits$ parts, and send to
the referee independent samples from the $L$-ary distribution $\q$
induced by $\p$ on this partition. For a random balanced partition
(i.e., where every part has cardinality $\ab/L$), clearly the uniform
distribution $\uniform_\ab$ is mapped to the uniform
distribution $\uniform_L$. Thus, one can hope to reduce the problem of
testing uniformity of $\p$ (over $[\ab]$) to that of testing
uniformity of $\q$ (over $[L]$). The latter task would be easy to
perform, as every player can simulate one sample from $\q$ and 
communicate it fully to the referee with $\log L = \numbits$ bits of
communication. Hence, the key issue is to argue that this random
``flattening'' of $\p$ would somehow preserve the distance to
uniformity; namely, that if $\p$ is $\eps$-far from $\uniform_\ab$,
then (with a constant probability over the choice of the random
partition) $\q$ will remain $\eps'$-far from $\uniform_L$, for some
$\eps'$ depending on $\eps$, $L$, and $\ab$. If true, then it is easy
to see that this would imply a very simple protocol with
$O(\sqrt{L}/{\eps'}^2)$ players, where all agree on a random
partition and send the induced samples to the referee, who then
runs a centralized uniformity test. Therefore, in order to apply the
aforementioned natural recipe, it suffices to derive a ``random
flattening'' structural result for $\eps' \asymp \sqrt{(L/\ab)}\eps$. 
 
An issue with this approach, unfortunately, is that the total variation
distance (that is, the $\lp[1]$ distance) does not behave as desired under
these random flattenings, and the validity of our desired result
remains unclear. Fortunately, an
analogous statement with respect to the $\lp[2]$ distance turns out to
be much more manageable and suffices for our purposes. In more
detail, we show that a random flattening of $\p$ does preserve, with
constant probability, the $\lp[2]$ distance to uniformity; in our
case, by Cauchy--Schwarz the original $\lp[2]$ distance will be at least 
$\gamma\asymp\eps/\sqrt{\ab}$, which implies using known $\lp[2]$
testing results that one can test uniformity of the ``randomly
flattened'' $\q$ with
$O(1/(\sqrt{L}\gamma^2))=O(\ab/(2^{\numbits/2}\eps^2))$ samples. This
yields the desired guarantees on the protocol. However, the proposed
algorithm suffers one drawback: The amount of public randomness required for the players to
agree on a random balanced partition is $\Omega(\ab\log L)
= \Omega(\ab\cdot\numbits)$, which in cases with large alphabet size $\ab$ can be prohibitive.\smallskip

This leads us to our second protocol, whose main advantage is
that it requires much fewer bits of randomness
(specifically, $O_{\eps}(2^\numbits\log\ab)$); however, this comes at
the price of some loss in simplicity. In fact, our second algorithm
too pursues a natural, perhaps more greedy, approach: Pick uniformly at random a subset $S\subseteq[\ab]$ of
size $s\eqdef 2^\numbits-1$ and 
communicate to the referee either an element in $S$ that equals the
observed sample or indicate that the sample does not lie in $S$. If
$\p$ is indeed uniform, then the probability $\p(S)$ of set $S$ satisfies
$\p(S)=s/\ab$ and the conditional distribution $\p^S$ given that the
sample lies in $S$ is
uniform. On the other hand, it is not difficult to show that if $\p$
is $\eps$-far from uniform in total variation distance, then the
expected contribution of elements in $S$ to the $\lp[1]$ distance of $\p$ to
uniform is order $\eps$. By an averaging 
argument, this implies that  with probability at least $\eps$ either
(i)~$\p(S)$ differs from $s/\ab$ by a $(1\pm\Omega(\eps))$ factor, or
(ii)~$\p_S$ is itself $\Omega(\eps)$-far from uniform.  

For a given $S$, detecting if (i)~holds requires roughly $\ab/(s\eps^2)$
samples (and hence as many players), while under (ii)~one would need
$(\ab/s)\cdot \sqrt{s}/\eps^2 = \ab/(\sqrt{s}\eps^2)$ players (the cost of rejection sampling, times
that of uniformity testing on support size $s$) to test
uniformity. When public randomness is available, the players can 
 choose jointly the same random set $S$, so this protocol is
valid. But there is a caveat. Since each choice of $S$ is
only ``good'' with probability $\eps$, to achieve a constant probability
of success one needs to repeat the procedure outlined above for
$\Omega(1/\eps)$ different choices of $S$. This, along with the
overhead cost of a union bound over all repetitions, leads to a bound of
$O({ {\ab}/({2^{\numbits/2}\eps^3})\cdot\log(1/\eps) })$
on the number of players~--~far from the optimal answer of
$\ns=O( \ab/(2^{\numbits/2}\eps^2) )$. To avoid the extra
$1/\eps$, we rely instead on \emph{Levin’s work investment strategy}
(see e.g.~\cite[Appendix A.2]{Goldreich:14}), which by a more careful
accounting enables us to avoid paying the cost of the naive averaging
argument. Instead, by considering logarithmically many different
possible ``scales'' $\eps_j$ of distance between $\p_S$ and uniform,
each with its own probability $\alpha_j$ of occurring for a random
choice of $S$ and by keeping track of the various costs that ensue,
we can get rid of this extra $1/\eps$ factor. This only leaves us
with an extra $\log(1/\eps)$ factor to handle, which arises due to the
union bound. To omit this extra factor, we refine our argument by
allocating different failure probabilities to every different test
conducted, depending on the respective scale used. By choosing
these probabilities so that their sum is bounded by a constant (for
instance, by setting $\delta_i \propto 1/j^2$), we can still ensure
overall correctness  with a high, constant probability, while the
extra cost $\log(1/\delta_j)$ for the $j$-th scale considered is
subsumed in the accounting using Levin's strategy. This finally 
yields the desired bound of $\ns=O( \ab/(2^{\numbits/2}\eps^2) )$
for the number of players.\smallskip

For the lower bound, we take recourse to Le Cam's two-point
method. Specifically, we use the construction proposed by
Paninski~\cite{Paninski:08} for proving the lower bound for sample complexity 
in the collocated setting. Roughly, we consider the problem
of distinguishing the uniform distribution from a randomly selected
element of the family of distributions consisting each of a perturbation of
uniform distribution where the probabilities conditioned on pairs of
consecutive elements are changed from unbiased coins to coins of bias
$\eps$. However, Paninski's original treatment does not suffice now as we need to
handle the total variation distance between the distribution induced on the message sequence $M$ 
under the uniform distribution and a uniform mixture of the pairwise perturbed distributions.
This is further bounded above by the average distance between the message distribution under uniform input and under the pairwise perturbed input.
In fact, treating public randomness as a common observation in both settings, it suffices to obtain a worst-case bound for deterministic inference protocols. 
Capitalizing on the fact that both distributions of messages are in this case product distributions, we can show that this average distance for deterministic protocols is bounded above by $\sqrt{\ns(2^\numbits\eps^2)/\ab}$, which leads to a lower bound of
$\ns=\Omega(\ab/(2^\numbits\eps^2))$. 

The bound obtained above is tight for $\numbits=1$, but is sub-optimal in general .  
To refine this bound further, instead of considering the average
distance between the distributions, we need to carefully analyze the
distance of uniform from the average. However, this quantity is not
amenable to standard bounds for total variation distance in terms of
Kullback--Leibler divergence and Hellinger distance devised to handle product
distributions, as the average ``\no-distribution'' is not itself a product distribution anymore. Instead, we take recourse to a technique used 
in~\cite{Paninski:08}, building on~\cite{Pollard:2003}, that uses a
$\chi^2$-distance bound and proceeds by expanding the product
likelihood ratios in multilinear form. Obtaining the final bound requires a
sub-Gaussian bound for a log-moment generating function, which is
completed by using a standard transportation method technique.

\subsection{Related prior work}
For clarity, we divide our discussion of the relevant literature into 
three parts: The first discussing the literature in the collocated setting, and the next two the
prior work concerned with distributed inference and simulation.

\paragraph{Inference in the collocated setting.} 
The goodness-of-fit problem is a classic hypothesis testing problem
with a long line of work in statistics, but the finite-alphabet variants of interest to us were
first considered by Batu et al.~\cite{BFRSW:00} and Goldreich,
Goldwasser, and Ron~\cite{GGR:98} under {distribution testing}, which
in turn evolved as a a branch of
{property testing}~\cite{RS:96,GGR:98}, a field of theoretical
computer science focusing on ``ultra-fast'' (sublinear-time)
algorithms for decision problems. 
 Distribution testing has received much attention in the past decade,
with considerable progress made and tight answers obtained for many
distribution properties (see
e.g.~surveys~\cite{Rubinfeld:12:Survey,Canonne:15:Survey,BW:17} and
references within for an overview). Most pertinent to our work is  {uniformity testing}~\cite{GRexp:00,Paninski:08,DGPP:17}, {the} prototypical distribution testing problem with applications to many other property testing problems~\cite{BKR:04,DKN:15,Goldreich:16,CDGR:17:journal}. 

Another inference question in the finite-alphabet setting that has
received a lot of attention in recent years is that of {functional
estimation}, where the goal is to estimate a function of the
underlying distribution.
Recent advances in this area have pinpointed the optimal rates for
functionals such as entropy, support size, and many others (see for
instance~\cite{Paninski:04,RRSS:09,ValiantValiant:11,JVW:14,WY:14,AOST:17,JVHW:17,ADOS:17}  for some of the most recent work).   

The extreme case of functional estimation is the fundamental question
of distribution learning, namely the classic density estimation
problem in statistic where the goal is estimate the entire distribution. With more than a century of history (see the books~\cite{Tsybakov:09, DL:01}), distribution learning has recently seen a surge of interest in the computer science community as well, with a focus on discrete domains (see e.g.~~\cite{Diakonikolas:CRC} for a survey of these recent developments).

\paragraph{Inference in the distributed setting.}
As previously mentioned, distributed hypothesis testing and estimation
problems were first studied in information theory, albeit in a
different setting than what we consider~\cite{AhlCsi86, Han87,
HanAmari98}. The focus in that line of work has been to characterize
the trade-off between asymptotic error exponent and communication rate
per sample. Recent extensions have considered interactive
communication~\cite{XiangKimISIT13}, more complicated communication
models~\cite{WiggerTimo16}, and even no
communication~\cite{Watanabe17}. The communication
complexity of independence testing for fixed error has been considered
recently in~\cite{SahasTyagiISIT:18}.

Closer to our work is distributed parameter estimation and functional
estimation that has gained significant attention in recent years (see
e.g.~\cite{DJW:13,GMN:14,GMNW:16,Watson:18}). In these works, much
like our setting, independent samples are distributed across players,
which deviates from the information theory setting described above
where each player observes a fixed dimension of each independent
sample. However, the communication model in these results differs from
ours, and the communication-starved regime we consider has not been
studied in these works.

Our communication model is the same as that considered
in~\cite{HOW:18}, which establishes, under some mild assumptions, a
general lower bound for estimation of model parameters under squared
$\lp[2]$ loss. Although the problems considered in our work differ
from those in~\cite{HOW:18} and the results are largely incomparable,
we build on a result of theirs to establish one of our lower
bounds.\footnote{In fact, the same communication model was proposed in
a different work, presented at the 2018 ITA
workshop~\cite{HMOW:18}. In this talk, the authors described a
protocol for learning discrete distributions under $\lp[1]$ error,
with a number of players that they showed to be optimal up to constant
factors.}

The problem of distributed density estimation, too, has gathered
recent interest in various statistical
settings~\cite{BPCPE:11,BBFM:12,ZDJW:13,Shamir:14,DGLNOS:17,HOW:18,XR:18,ASZ:18}. Our
work is closest to two of these: The aforementioned~\cite{HOW:18,
HMOW:18} and~\cite{DGLNOS:17}. The latter considers both $\lp[1]$
(total variation) and $\lp[2]$ losses, although in a different setting
than ours. Specifically, they study an interactive model where the
players do not have any individual communication constraint, but
instead the goal is to bound the {total} number of bits communicated
over the course of the protocol. This difference in the model leads to
incomparable results and techniques (for instance, the lower bound for
learning $\ab$-ary distributions in our model is higher than the upper
bound in theirs).

Our current work further deviates from this prior literature, since we
consider distribution testing as well and examine the role of
public-coin for SMPs. Additionally, a central theme here is the
connection to distribution simulation and its limitation in enabling
distributed testing. In contrast, the prior work on distribution
estimation, in essence, establishes the optimality of simple protocols
that rely on distributed simulation for inference. (We note that
although recent work of~\cite{BCG:17} considers both communication
complexity and distribution testing, their goal and results are very
different~--~indeed, they explain how to leverage on negative results
in the standard SMP model of communication complexity to obtain sample
complexity lower bounds in {collocated} distribution testing.)

\paragraph{Distributed simulation.} Problems related to joint
simulation of probability distributions have been the object of focus
in the information theory and computer science literature. 
Starting with the works of G{\'a}cs and K{\"o}rner~\cite{GK:73} and
Wyner~\cite{Wyner:75} where the problem of generating shared
randomness from correlated randomness and vice-versa, respectively,
were considered, several important variants have been studied such as
{correlated sampling}~\cite{Broder:97,KT:02,Holenstein:07,BGHKRS:16}
and {non-interactive simulation}~\cite{KA:12,GKS:16,DMN:18}.
Yet, our problem of {exact} simulation of a single (unknown)
distribution with communication constraints from multiple parties has
not been studied previously to the best of our knowledge.

\subsection{Organization}
We begin by setting notation and recalling some useful definitions and
results in~\cref{sec:preliminaries}, before formally introducing our
distributed model in~\cref{sec:model}. \cref{sec:simulation}
introduces the question of distributed simulation and contains our
protocols and impossibility results for this problem
(specifically,~\cref{theo:distributed:simulation:impossiblity} 
and~\cref{theo:distributed:simulation} are proven
in~\cref{ssec:sampling:impossibility} and
in~\cref{ssec:sampling:possibility}). In~\cref{sec:simulation:inference}, we consider the
relation between distributed simulation and (private-coin)
distribution inference. Namely, we explain
in~\cref{sec:sampling:applications} how a distributed simulation
protocol immediately implies protocols for every inference task
(\cref{theo:inference:simulation:informal}) and instantiate this result for two concrete examples of distribution learning and uniformity
testing. \cref{sec:conjecture} is concerned with
~\cref{conj:flying:poney}: ``Is inference {via} distributed simulation optimal in general?'' After rigorously formalizing this question, we answer it in the negative in~\cref{theo:refuting:fpc}.

The subsequent section,~\cref{sec:uniformity}, focuses on the problem
of uniformity testing and contains the proofs of the upper and lower
bounds of~\cref{theo:uniformity:shared:randomness} (as
previously mentioned, we provide there two proofs of the upper bound
using different protocols, with a simple, albeit randomness-heavy, protocol, and a more involved,
randomness-savvy one).

Although we rely throughout on the formal description of our model
given in~\cref{sec:model}, the other sections are self-contained and
can be read independently.

\section{Preliminaries}\label{sec:preliminaries}
We write $\log$ (resp. $\ln$) for the binary (resp. natural) logarithm, and $[\ab]$ for the set of integers $\{1,2,\dots,\ab\}$. Given a fixed (and known) discrete domain $\domain$ of size $\ab$, we denote by $\distribs{\domain}$ the set of probability distributions over $\domain$, i.e.,
  \[
      \distribs{\domain} = \setOfSuchThat{ \p\colon\domain\to[0,1] }{ \normone{\p}=1 }\,.
  \]
  A \emph{property of distributions} over $\domain$ is a subset $\property \subseteq\distribs{\domain}$. Given $\p\in\distribs{\domain}$ and a property $\property$, the distance from $\p$ to the property is defined as
  \begin{equation}
      \totalvardist{ \p }{\property} \eqdef \inf_{\q\in\property} \totalvardist{ \p }{ \q }
  \end{equation}
  where $\totalvardist{ \p }{ \q } = \sup_{S\subseteq\domain} \left(\p(S)-\q(S)\right)$ for $\p,\q\in\distribs{\domain}$, is the \emph{total variation distance} between $\p$ and $\q$. For a given parameter $\eps\in(0,1]$, we say that $\p$ is \emph{$\eps$-close} to $\property$ if $\totalvardist{ \p }{\property}\leq \eps$; otherwise, we say that $\p$ is \emph{$\eps$-far} from $\property$. For a discrete set $\domain$, we write $\uniformOn{\domain}$ for the uniform distribution on $\domain$, and will sometimes omit the subscript when the domain is clear from context. We indicate by $x\sim\p$ that $x$ is a sample drawn from the distribution $\p$.
  
In addition to total variation distance, we shall rely in some of our proofs on the $\chi^2$ and Kullback--Leibler (KL) divergences between discrete distributions $\p,\q\in\distribs{\domain}$, defined respectively as $\chi^2(\p,\q) := \sum_{x\in\domain} \frac{(\p_x-\q_x)^2}{\q_x(1-\q_x)}$ and $D(\p\|\q) := \sum_{x\in\domain} \p_x \ln\frac{\p_x}{\q_x}$.
  
We use the standard asymptotic notation $\bigO{\cdot}$, $\bigOmega{\cdot}$, and $\bigTheta{\cdot}$; and will sometimes write $a_n \lesssim b_n$ to indicate that there exists an absolute constant $c>0$ such that $a_n \leq c\cdot b_n$ for all $n$. Finally, we will denote by $a\wedge b$ and $a\vee b$ the minimum and maximum of two number $a$ and $b$, respectively. 
 
\section{Communication, Simulation, and Inference Protocols}\label{sec:model}
We set the stage by describing the communication protocols we study for both the distributed simulation and the distributed inference problems. 
Throughout the paper, we restrict to simultaneous communication models with private and public randomness. 

Formally, $n$ players observe samples $X_1,\dots, X_{\ns}$ with player $i$ given access to $X_i$. The samples are assumed to be generated independently from an unknown distribution $\p$.  In addition, player $i$ has access to uniform randomness $U_i$ such that $(U_1, \ldots, U_{\ns})$ is jointly independent of $(X_1,\dots,X_{\ns})$. An $\numbits$-bit \emph{simultaneous message-passing} (SMP) communication protocol $\pi$ for the players consists of $\{0,1\}^{\numbits}$-valued mappings $\pi_1, \dots, \pi_{\ns}$ where player $i$ sends the message $M_i=\pi_i(X_i, U_i)$. 
The message $M=(M_1, \dots, M_{\ns})$ sent by the players is received by a common referee. Based on the assumptions on the availability of the randomness $(U_1, \dots, U_{\ns})$ to the referee and the players, three natural classes of protocols arise:
\begin{enumerate}
  \item 
  \textit{Private-coin protocols:} $U_1,  \dots, U_{\ns}$ are mutually independent and  unavailable to the referee.
  \item 
  \textit{Pairwise-coin protocols:} $U_1, \dots, U_{\ns}$ are mutually independent and  available to the referee.
  \item
   \textit{Public-coin protocols:} All player and the referee have access to $U_1, \dots, U_{\ns}$. 
\end{enumerate} 
In this paper, we focus only on private- and public-coin communication protocols; an interesting question is distinguishing pairwise-coin protocols from the other two. For the ease of presentation, we represent the private randomness
communication $f_i(x_i, U_i)$ using a channel $W_i\colon \cX \to \{0,1\}^{\numbits}$
where player $i$ upon observing $x_i$ declares $y$ with probability
$W_i(y|x_i)$. Also, for public-coin protocols, we can assume without loss
of generality that $U_1 = U_2 = \dots = U_{\ns} = U$. 
 
\paragraph*{Distributed simulation protocols.} An $\numbits$-bit \emph{simulation} $\cS = (\pi, \delta)$ of $\ab$-ary distributions using $\ns$ players consists of an $\numbits$-bit SMP $\pi$ and a decision map $\delta$ comprising mappings $\delta_x\colon (M,U) \mapsto [0,1]$ such that for each message $m$ and randomness $u$, 
\[
\sum_x\delta_x(m,u)\leq 1. 
\] 
Upon observing the message $M=(M_1, \dots, M_{\ns})$ and (depending on the type of protocol) randomness $U=(U_1,
\dots, U_{\ns})$, the referee declares the random sample $\hat X = x$ with
probability $\delta_x(M, U)$ or declares an abort symbol $\perp$ if no
$x$ is selected. For concreteness, we assume that the random variable
$\hat X$ takes values in $\cX\cup\{\perp\}$ with $\{\hat X=\perp\}$
corresponding to the abort event. When $\pi$ is a private, pairwise, or public-coin protocol, respectively, the simulation $\cS$ is called private, pairwise, or public-coin simulation.\smallskip

\noindent A simulation $\cS$ is an \emph{$\alpha$-simulation} if for every $\p$
\[
\probaDistrOf{\p}{\hat X =x \mid \hat X \neq \perp} = \p_x, \quad \forall\, x\in \cX,
\]
and the abort probability satisfies 
\[
\probaDistrOf{\p}{\hat X =\perp} \leq \alpha.
\] 
When the
probability of abort is \emph{zero}, $\cS$ is termed a \emph{perfect
simulation}.
   
\paragraph*{Distributed inference protocols.} We give a general
definition of distributed inference protocols that is applicable
beyond the use-cases considered in this work. An inference
problem $\task$ can be described by a 
tuple $(\class, \cX, \cE, L)$ where $\class$ 
denotes a family of distributions on the alphabet $\cX$, $\cE$ a class
of allowed estimates for elements of $\class$ (or their functions), and 
$L\colon \class\times\cE \to \R_+^q$ is a loss function that evaluates the accuracy of our
estimate $e\in \cE$ when $\p\in \class$ was the ground truth. 

An $\numbits$-bit \emph{distributed inference protocol} $\cI =
(\pi, e)$ for the
inference problem $(\class, \cX, \cE, L)$ consists of an $\numbits$-bit SMP $\pi$ and an estimator $e$
available to the referee who, upon observing the message $M = \pi(X^n,
U)$ and the randomness $U$, estimates the unknown $\p$ as $e(M, U)\in \cE$. As
before, we say that a private, pairwise, or public-coin inference protocol, respectively, uses a private, pairwise, or public-coin communication protocol $\pi$.

For $\vec{\gamma}\in\R_+^q$, an inference protocol $(\pi, e)$ is a
\emph{$\vec{\gamma}$-inference protocol} if 
\[
\bE{\p}{L_i(\p, e(M, U))} \leq \gamma_i, \quad\, \forall 1\leq i \leq q.
\]
We instantiate the abstract definition above in two illustrative
examples that we will pursue in this paper.

\begin{example}[Distribution learning] Consider  the
problem $\Learning_{\ab}(\eps,\delta)$ of estimating a
$\ab$-ary distribution $\p$ by observing independent samples from it,
namely the finite alphabet distribution learning problem. This problem
is obtained from the general formulation above by setting $\cX$
to be $[\ab]$, $\class$ and $\cE$ both to be the $(\ab-1)$-dimensional
probability simplex $\class_\ab$, and $L(\p, \hat{\p})$ as follows:
\[
L(\p, \hat{\p}) = \indic{\totalvardist{\p}{\hat {\p}}> \eps}.
\]
For this case, we term the $\delta$-inference protocol an $\numbits$-bit
$(\ab, \eps, \delta)$-\emph{learning protocol} for $\ns$ player. 
\end{example}

\begin{example}[Uniformity testing] In the uniformity testing problem
$\Testing_{\ab}(\eps, \delta)$, our goal is to determine whether $\p$ is
the uniform distribution $\uniformOn{\ab}$ over $[\ab]$ or if it
satisfies $\totalvardist{\p}{\uniformOn{\ab}} >\eps$. This can be
obtained as a special case of our general formulation by setting
$\cX = [\ab]$, $\class$ to be set containing $\uniformOn{\ab}$ and all
$\p$ satisfying $\totalvardist{\p}{\uniformOn{\ab}} >\eps$, $\cE
= \{0,1\}$, and loss function $L$ to be
\[
L(\p, b) = b\cdot \indic{\p= \uniformOn{\ab}}+
(1-b)\cdot \indic{\p\neq \uniformOn{\ab}}, \quad b \in \{0,1\}.
\]
For this case, we term the $\delta$-inference protocol an $\numbits$-bit
$(\ab, \eps, \delta)$-\emph{uniformity testing protocol} for $\ns$ players. Further, for simplicity we will refer to $(\ab, \eps, 1/3)$-uniformity testing protocols simply as $(\ab, \eps)$-\emph{uniformity testing protocols}.
\end{example}

Note that distributed variants of several other inference problems
such as that of estimating functionals of distributions and parametric
estimation problems can be included as instantiations of the
distributed inference problem described above. 

We close by noting that while we have restricted to the SMP model of
communication, the formulation can be easily extended to include
interactive communication protocols where the communication from each
player can be heard by all the other players (and the referee), and in
its turn, a player communicates using its local observation and the
communication received from all the other players in the past. A formal
description of such a protocol can be given in the form of a multiplayer
protocol tree \emph{\`a la} \cite{KushilevitzNisan97}. However, such
considerations are beyond the scope of this paper.

\paragraph*{A note on the parameters.} It is immediate to see that for $\numbits \geq \log \ab$ the distributed and centralized settings are equivalent, as the players can simply send their input sample to the referee (thus, both upper and lower bounds from the centralized setting carry over).

\section{Distributed Simulation}\label{sec:simulation}
In this section, we consider the distributed simulation problem described in the previous section. We start by considering the more ambitious problem of perfect simulation and 
show that when $\numbits< \log \ab$, perfect simulation using $\ns$ players is impossible using any $\ns$. Next, we consider $\alpha$-simulation for a constant $\alpha\in(0,1)$ and 
exhibit an $\numbits$-bit $\alpha$-simulation of $\ab$-ary distributions using $O({\ab}/{2^\numbits})$ players. In fact, by drawing on a reduction from distributed distribution learning, we will show in~\cref{sec:sampling:applications} that this is the least number of players required (up to a constant factor) for $\alpha$-simulation for any $\alpha\in(0,1)$.

\subsection{Impossibility of perfect simulation when $\numbits<\log \ab$}
\label{ssec:sampling:impossibility}
We begin with a proof of impossibility which shows that any simulation that works for all points in the interior of 
the $(\ab-1)$-dimensional probability simplex must fail for a distribution on the boundary.
Our main result of this section is the following:
\begin{theorem}\label{theo:sampling:impossibility:non:adaptive:k:l}
For any $n\geq 1$, there exists no $\numbits$-bit perfect simulation
of $\ab$-ary distributions using $\ns$ players unless
$\numbits \geq \log \ab$. 
\end{theorem}

\begin{proof}
Let $\cS=(\pi, \delta)$ be an $\numbits$-bit perfect simulation for
$\ab$-ary distributions using $\ns$ players. Suppose that
$\numbits< \log \ab$. We show a contradiction for any such public-coin
simulation $\cS$. Fix a realization $U=u$ of the public randomness. By
the pigeonhole principle we can find a message vector $m=(m_1, \dots, m_n)$ and distinct elements $x_i,
x_i'\in [\ab]$ for each $i\in [n]$ such that 
\[
\pi_i(x_i, u) = \pi_i(x_i',u) = m_i.
\]
Note that the probability of declaring $\bot$ for a public-coin simulation
must be $0$ for every $\ab$-ary distribution. Therefore, since the
message $m$ occurs with a positive probability under a distribution
$\p$ with $\p_{x_i}>0$ for all $i$, the referee must declare an output $x\in [\ab]$ with 
positive probability when it receives $m$, i.e., there exists $x\in
[\ab]$ such that $\delta_x(m, u)>0$. Also, since $x_i$ and $x_i'$ are
distinct for each $i$, we can assume without loss of generality that
$x_i\neq x$ for each $i$.  Now, consider a distribution $\p$ such that
$\p_x = 0$ and $\p_{x_i}> 0$ for each $i$. For this case, the referee
must never declare $\p_x$, i.e., $\probaOf{\hat X =x} = 0$. In
particular, $\probaCond{\hat X=x }{ U=u}$ must be $0$, which can only happen
if $\probaCond{M=m}{U=u}=0$. But since $\p_{x_i}>0$ for each $i$, 
\[
\probaCond{M=m}{U=u} \geq \prod_{i=1}^n\p_{x_i} > 0\,,
\]
which is a contradiction.
\end{proof}
Note that the proof above shows, as stated before, that any perfect simulation that
works for every $\p$ in the interior of the $(\ab-1)$-dimensional
probability simplex, must fail at one point on the boundary of the
simplex. In fact, a much stronger impossibility result holds. We show
next that for $\ab = 3$ and $\numbits=1$, we cannot find a perfect simulation
that works in the neighborhood of any point in the interior of the
simplex. 

\begin{theorem}\label{theo:sampling:impossibility:non:adaptive:k3}
For any $\ns\geq 1$, there does not exist any $\numbits$-bit perfect
simulation of $3$-ary distributions unless $\numbits\geq 2$, even
under the promise that the input distribution comes from an open set
in the interior of the probability simples.
\end{theorem}

Before we prove the theorem, we show that there is no loss of
generality in restricting to~\emph{deterministic} protocols, namely
protocols where each player uses a deterministic function of its
observation to communicate. The high-level argument is relatively simple: By
replacing player $j$ by two players $j_1,j_2$, each with a suitable
deterministic strategy, the two $1$-bit messages received by the
referee will allow him to simulate player $j$'s original randomized
mapping. 

\begin{lemma}\label{lemma:sampling:impossibility:non:adaptive:deterministic}
For $\cX=\{0, 1, 2\}$, suppose there exists a $1$-bit perfect
simulation $S'=(\pi', \delta')$ with $\ns$ players. Then, there is a $1$-bit perfect
simulation $S=(\pi,\delta)$ with $2\ns$ players such that, for each
$j\in [2\ns]$, the communication $\pi$ is deterministic, i.e., for
each realization $u$ of public randomness
\[
\pi_j(x_j, u) =\pi_j(x), \qquad x\in \cX\,.
\]
\end{lemma}
\begin{proof}
Consider the mapping $f\colon\{0,1,2\}\times \{0,1\}^\ast\to \{0,1\}$. We will
show that we can find mappings $g_1\colon \{0,1,2\}\to \{0,1\}$,
$g_2\colon \{0,1,2\}\to \{0,1\}$, and
$h\colon \{0,1\}\times \{0,1\} \times \{0,1\}^\ast\to \{0,1\}$ such that for every
$u$ 
\begin{equation}\label{e:deterministic}
\probaOf{f(X, u)=1} = \probaOf{h(g_{1}(X_1), g_{2}(X_2),u)=1},
\end{equation}
where random variables $X_1$, $X_2$, $X$  are independent and
identically distributed and take
values in $\{0,1,2\}$. We can then use this construction to get our claimed simulation $S$
using $2\ns$ players as follows: Replace the communication $\pi_j'(x,u)$ from  player
$j$ with communication  $\pi_{2j-1}(x_{2j-1})$ and
$\pi_{2j}(x_{2j})$, respectively, from two players
$2j-1$ and $2j$, where $\pi_{2j-1}$ and $\pi_{2j}$ correspond to
mappings $g_1$ and $g_2$ above for $f= \pi'_j$. The referee can then
emulate the original protocol using the corresponding mapping $h$ and
using $h(\pi_{2j-1}(x_{2j-1}), \pi_{2j}(x_{2j}), u)$ in place of communication
from player $j$ in the original protocol. Then, since the probability
distribution of the communication does not change, we retain the
performance of $S'$, but using only deterministic communication now.

Therefore, it suffices to establish \eqref{e:deterministic}. For convenience, denote $\alpha_u\eqdef \indic{f(0,u) =
1}$, $\beta_u\eqdef \indic{f(1,u) = 1}$, and
$\gamma_u\eqdef \indic{f(2,u) = 1}$. Assume without loss of generality that
$\alpha_u\leq \beta_u+\gamma_u$; then,
$(\beta_u+\gamma_u-\alpha_u)\in\{0,1\}$. Let $g_i(x)=\indic{x=i}$ for
$i\in\{1,2\}$. Consider the mapping $h$ given by
\[
h(0,0,u)=\alpha_u,\,\, h(1,0,u)=\beta_u, \,\, h(0,1,u)=\gamma_u,\,\, h(1,1,u)=(\beta_u+\gamma_u-\alpha_u)\,.
\]
Then, for every $u$,
\begin{align*}
&\probaOf{h(g_1(X_1), g_2(X_2),
u)=1}
\\
&\qquad= \alpha_u(1-\p_1)(1-\p_2)+ \beta_u(1-\p_1)\p_2
+\gamma_u\p_1(1-\p_2)+(\beta_u+\gamma_u-\alpha_u)\p_1\p_2
\\
&\qquad= \alpha_u(1-\p_1-\p_2)+\beta_u\p_2+\gamma_u\p_1 = \probaOf{f(X,u) = 1}\,,
\end{align*}
which completes the proof.
\end{proof}
We now prove~\cref{theo:sampling:impossibility:non:adaptive:k3}, but
in view of our previous observation, we only need to consider deterministic communication. 
\begin{proofof}{\cref{theo:sampling:impossibility:non:adaptive:k3}}
Suppose by contradiction that there exists such a $1$-bit perfect
simulation protocol $S=(\pi,\delta)$ for $\ns$ players on
$\cX=\{0,1,2\}$ such that $\pi(x,u)=\pi(x)$. Assume that this protocol is correct for all distributions $\p$ in the neighborhood of some $\p^\ast$ in the interior of the simplex. Consider a 
partition the players into three sets $\cS_0$, $\cS_1$, and $\cS_2$, with
\[
    \cS_i \eqdef \setOfSuchThat{ j\in[\ns] }{ \pi_j(i) = 1 }, \qquad  i\in\cX\,.
\]
Note that for deterministic communication the message $M$ is
independent of public randomness $U$. Then, by the definition of perfect simulation, it must be the case that
\begin{align}\label{eq:lowerbound:sampling:output-probability}
\p_x &= \shortexpect_{U} \sum_{m\in\{0,1\}^{\ns}} \delta_x(m,U) \probaCond{ M = m }{ U } = \shortexpect_{U} \sum_{m} \delta_x(m,U) \probaOf{ M = m } \notag\\
     &= \sum_{m} \shortexpect_U[\delta_x(m,U)] \probaOf{ M = m }
\end{align}
for every $x\in\cX$, which with our notation of $\cS_0, \cS_1, \cS_2$
can be re-expressed as
\begin{align*}\label{eq:lowerbound:sampling:output-probability}
\p_x &= \sum_{m\in\{0,1\}^{\ns}} \shortexpect_U[\delta_x(m,U)] \prod_{i=0}^2 \prod_{j\in\cS_i} (m_j \p_i + (1-m_j)(1-\p_i)) \\
    &= \sum_{m\in\{0,1\}^{\ns}} \shortexpect_U[\delta_x(m,U)] \prod_{i=0}^2 \prod_{j\in\cS_i} (1-m_j +(2m_j-1)\p_i)\,,
\end{align*}
for every $x\in\cX$. But since the right-side above is a polynomial in
$(\p_0,\p_1, \p_2)$, it can only be zero in an open set in the interior
if it is identically zero. In particular, the constant term must be zero:
\begin{align*}
    0
    = \sum_{m\in\{0,1\}^{\ns}} \shortexpect_U[\delta_x(m,U)] \prod_{i=0}^2 \prod_{j\in\cS_i} (1-m_j)
    = \sum_{m\in\{0,1\}^{\ns}} \shortexpect_U[\delta_x(m,U)] \prod_{j=1}^\ns
    (1-m_j)\,.
\end{align*}
Noting that every summand is non-negative, this implies that for all
$x\in\cX$ and $m\in\{0,1\}^\ns$,
$\shortexpect_U[\delta_x(m,U)] \prod_{j=1}^\ns (1-m_j) = 0$. In
particular, for the all-zero message $\textbf{0}^\ns$, we get
$ \shortexpect_U[\delta_x(\textbf{0}^\ns,U)] = 0 $ for all $x\in\cX$,
so that again by non-negativity we must have
$\delta_x(\textbf{0}^\ns,u)=0$ for all $x\in\cX$ and randomness
$u$. But the message $\textbf{0}^\ns$ will happen with probability
\[
\probaOf{M=\textbf{0}^\ns} = \prod_{i=0}^2 \prod_{j\in\cS_i} (1-\p_i)
= (1-\p_0)^{\abs{\cS_0}}(1-\p_1)^{\abs{\cS_1}}(1-\p_2)^{\abs{\cS_2}} > 0,
\]
where the inequality holds since $\p$ lies in the interior of the
simplex. Therefore, for the output $\hat{X}$ of the referee we have
\begin{align*}
    \probaOf{ \hat{X} \neq \bot }
    &= \sum_m \sum_{x\in \cX} \shortexpect_U[\delta_x(m,U)]\cdot \probaOf{M=m}
    = \sum_{m\neq \textbf{0}^\ns} \probaOf{M=m} \sum_{x\in \cX} \shortexpect_U[\delta_x(m,U)] \\
    &\leq \sum_{m\neq \textbf{0}^\ns} \probaOf{M=\textbf{0}^\ns} =
    1- \probaOf{M=\textbf{0}^\ns} < 1\,,
\end{align*}
contradicting the fact that $\pi$ is a perfect simulation protocol.
\end{proofof}

\begin{remark}
It is unclear how to extend the proof
of~\cref{theo:sampling:impossibility:non:adaptive:k3} arbitrary
$\ab, \numbits$. In particular, the proof of~\cref{lemma:sampling:impossibility:non:adaptive:deterministic}
does not extend to the general case. A plausible proof-strategy is a 
black-box application of the $\ab=3$, $\numbits=1$ result 
to obtain the general result  using a direct-sum-type argument.
\end{remark}

We close this section by noting that perfect
simulation is impossible even when the communication from each player
is allowed to depend on that from the previous ones.  Specifically,
we show that availability of such an interactivity can at most bring an
exponential improvement in the number of players.
\begin{lemma}
  For every $\ns\geq 1$, if there exists an interactive public-coin
$\numbits$-bit perfect simulation of $\ab$-ary distributions with
$\ns$ players, then there exists a
public-coin $\numbits$-bit perfect simulation of $\ab$-ary
distributions with $2^{\numbits\ns+1}$ players that uses only SMP.
\end{lemma}
\begin{proof}
Consider an interactive communication protocol $\pi$ for
distributed simulation with $\ns$ players and
$\numbits$ bits of communication per player.  We can view the overall
protocol as a $(2^{\numbits})$-ary tree of depth $\ns$ where
player $j$ is assigned all the nodes at depth $j$. An execution of the protocol is a path
from the root to the leaf of the tree. Suppose the protocol starting
at the root has reached a node at depth $j$, then the next node at
depth $j+1$ is determined by the communication from player $j$. 
Thus, this protocol can be simulated non-interactively using at most
$ ((2^{\numbits})^{\ns}-1)/(2^\numbits-1) <
        2^{\numbits\ns+1}$ players, where players $(2^{j-1}+1)$ to $2^j$ send all messages
correspond to nodes at depth $j$ in the tree. Then, the
referee receiving all the messages can output the leaf by following
the path from root to the leaf.
\end{proof}

\begin{corollary}
\cref{theo:sampling:impossibility:non:adaptive:k:l,theo:sampling:impossibility:non:adaptive:k3} extend to interactive protocols as well.
\end{corollary}

\subsection{An $\alpha$-simulation protocol using rejection sampling}
\label{ssec:sampling:possibility}

In this section, we establish~\cref{theo:distributed:simulation} and
provide  $\alpha$-simulation protocols for $\ab$-ary distributions
using   $\ns=O(\ab/2^\numbits)$ players. We first present the protocol
for the case $\numbits=1$, before extending it to general $\numbits$.
The proof of lower bound for the number of players required for
$\alpha$-simulation of $\ab$-ary distributions is based on the
connection between distributed simulation and distributed distribution
learning and will be provided in the next section where this
connection is discussed in detail.

For ease of presentation, we allow a slightly different class of
protocols where we have an infinitely long sequence of players, each
with access to one independent sample from the unknown $\p$. The
referee's protocol entails checking each player's message and deciding
either to declare an output ${\hat X=x}$ and stop, or see the next
player's output. We assume that with probability one the referee uses
finitely many players and declares an output. The cost of maximum
number of players of the previous setting is now replaced with 
the  expected number of players used to declare an output. By an
application  of Markov's inequality, this can be easily related to our
original setting of private-coin $\alpha$-simulation.  

\begin{theorem}\label{theo:generate:sample:1bit:1sample}
  There exists a $1$-bit private-coin protocol that outputs a sample
  $x\sim \p$ using messages of at most $20\ab$ players in expectation.
\end{theorem}
\begin{proof}
To help the reader build heuristics for the proof, we describe the
protocol and analyze its performance in steps. We begin by describing the basic idea and
building blocks; we then build upon it to obtain a full-fledged
protocol, but with potentially unbounded expected number of players
used. Finally, we
describe a simple modification which yields our desired bound for
expected number of player's accessed.

\paragraph{The scheme, base version.}
Consider a protocol with $2\ab$ players where the $1$-bit
communication from players $(2i-1)$ and $(2i)$ just indicates if their
observation is $i$ or not, namely $\pi_{2i-1}(x)=\pi_{2i}(x)=\indic{x=i}$.

\noindent On receiving these $2\ab$ bits, the referee~$\referee$ acts as follows:
\begin{itemize}
  \item if exactly one of the bits $
M_1, M_3, \dots, M_{2\ab-1}$ is equal to one, say the bit $M_{2i-1}$,
and the corresponding bit $M_{2i}$ is zero, then 
the referee outputs $\hat X= i$; \item otherwise, it outputs
  $\bot$.
\end{itemize}
In the above, the probability $\rho_{\p}$ that some $i\in[\ab]$ is
declared as the output (and not $\bot$) is 
\[
    \rho_{\p} \eqdef \sum_{i=1}^{\ab} \left(\p_i \prod_{j\neq i}
    (1-\p_j)\right) \cdot (1-\p_i) = \prod_{j=1}^{\ab}
    (1-\p_j) \cdot \sum_{i=1}^{\ab} \p_i = \prod_{j=1}^{\ab} (1-\p_j),
\]
so that
\[
    \rho_{\p} = \exp \sum_{j=1}^{\ab} \ln(1-\p_j)
    = \exp\mleft( -\sum_{t=1}^\infty \frac{\norm{\p}_t^t}{t} \mright)\geq \exp\mleft( -\mleft(1+\sum_{t=2}^\infty \frac{\normtwo{\p}^t}{t}\mright) \mright)
    = \frac{1-\normtwo{\p}}{e^{1-\normtwo{\p}}}
\]
which is bounded away from $0$ as long as $\p$ is far from being a
point mass.

Further, for any fixed $i\in[\ab]$, the probability that $\referee$
outputs $i$ is
\[
      \p_i\cdot \prod_{j=1}^{\ab} (1-\p_j) = \p_i\rho_{\p}\propto \p_i\,.
\]

\paragraph{The scheme, medium version.}
The (almost) full protocol proceeds as follows. Divide the countably
infinitely many players into successive, disjoint batches of $2\ab$
players each, and apply the base scheme to each of these runs. Execute
the base scheme to each of the batch, one at a time and moving to the next
batch only when the current batch declares a $\bot$; else declare the
output of the batch as $\hat X$.

It is straightforward to verify that the distribution of the output
$\hat X$ is exactly $\p$, and moreover that on expectation
$1/\rho_{\p}$ runs are considered before a sample is
output. Therefore, the expected number of players accessed
(i.e., bits considered by the referee) satisfies
\begin{equation}\label{eq:sampling:scheme:expected:query}
    \frac{2\ab}{\rho_{\p}} \leq
    2\ab \cdot \frac{e^{1-\normtwo{\p}}} {1-\normtwo{\p}}
\,.
\end{equation}

\paragraph{The scheme, final version.}
The protocol described above can have the expected number of players
blowing to infinity when  $\p$ has $\lp[2]$
norm close to one. To circumvent this difficulty, we modify the
protocol as follows: Consider the distribution $\q$ on $[2\ab]$
defined by
\[
    \q_{2i}=\q_{2i-1} = \frac{\p_i}{2},\qquad i\in[\ab]\,.
\] 
Clearly, $\normtwo{\q}=\normtwo{\p}/\sqrt{2}\leq
1/\sqrt{2}$, and therefore by~\eqref{eq:sampling:scheme:expected:query} the expected number of players required to
simulate $\q$ using our previous protocol is at most 
\[
4k\cdot \frac{e^{1-\frac 1{\sqrt{2}} }} {1-\frac 1{\sqrt 2}}
\leq 20k.
\]
But we can simulate a sample from $\p$ using a sample from $\q$ simply by
mapping $(2i-1)$ and $2i$ to $i$. The only thing remaining now is to
simulate samples from $\q$ using samples from $\p$. This, too, is
easy. Every $2$ players in a batch that declare $1$ on observing symbols $(2i-1)$ and
$(2i)$ from $\q$ declare $1$ when they see $i$ from $\p$. The referee
then simply flips each of this $1$ to $0$, thereby simulating the
communication corresponding to samples from $\q$. In summary, we
modified the original protocol for $\p$ by replacing each player with
two identical copies and modifying the referee to flip $1$
received from these players to $0$ independently with probability $1/2$; the output is declared in a
batch only when there is exactly one $1$ in the modified messages, in
which case the output is the element assigned to the player that sent $1$. 
Thus, we have a
simulation for $\ab$-ary distributions that uses at most $20\ab$
players, completing the proof of the theorem.
\end{proof}

Moving now to the more general setting, we have the following result. 

\begin{theorem}\label{theo:generate:sample:lbits:1sample}
  For any $\numbits\geq 2$, there exists a $\numbits$-bit private-coin
  protocol that outputs a sample  $x\sim \p$ using messages of at most
  $20\clg{\frac{\ab}{2^{\numbits}-1}}$ players in expectation.
\end{theorem}
\begin{proof}
For simplicity, assume that $2^\numbits -1$  divides $\ab$. We can then
 extend the previous protocol by considering a partition of
domain into $m=\ab/(2^{\numbits}-1)$ parts and assigning one part of 
size $2^\numbits-1$ each to a player. Each player then sends the all-zero
sequence of length $\numbits$ when it does not see an element from its
assigned set, or indicates the precise element from its assigned set
that it observed. For each batch, the referee, too, proceeds as before
 and declares an output if exactly one player in the batch sends a
 $1$ -- the declared output is the element indicated by the player
 that sent a $1$; else it moves to the next batch. To bound the number
 of players, consider the analysis of the base protocol. The
 probability that an output is declared for a batch (a $\bot$ is
not declared in the base protocol) is given by
\begin{align*}
    \rho_{\p} &\eqdef \sum_{i=1}^m \sum_{\numbits\in
    S_i} \left(\p_\numbits \prod_{j\neq i} (1-\p(S_j))\right) \cdot
    (1-\p(S_i)) \\
&= \prod_{j=1}^m
    (1-\p(S_j)) \cdot \sum_{i=1}^m \sum_{\numbits\in S_i} \p_\numbits \\
  &  = \prod_{j=1}^m (1-\p(S_j))\,,
\end{align*}
where $\{S_1, \dots, S_m\}$ denotes the partition used. Then, writing $\p^{(S)}$  for the distribution on $[m]$ given by $\p^{(S)}(j)
= \p(S_j)$, by proceeding as in the $\numbits=1$ case  we obtain
\[
\rho_\p\geq \frac{1-\normtwo{\p^{(S)}}}{e^{1-\normtwo{\p^{(S)}}}}\,.
\]
Once again, this quantity may be unbounded and we circumvent this
difficulty by replacing each player with two players that behave
identically and flipping their communicated $1$'s to $0$'s randomly at the
referee; the output is declared in a batch only when there is exactly
one $1$ in the modified messages, in
which case the output is the element indicated by the player that sent
$1$. The analysis can be completed exactly in the manner of the
$\numbits=1$ case proof by noticing that the protocol is tantamount to
simulating $\q$ with $\normtwo{\q^{(S)}}\leq 1/\sqrt{2}$ and accesses
messages from at
most $20m$ players in expectation.
\end{proof}
 
\section{Distributed Simulation for Distributed Inference}\label{sec:simulation:inference}

In this section, we focus on the connection between distributed
simulation and (private-coin) distributed inference. We first describe
the implications of the results from~\cref{sec:simulation}
for \emph{any} distributed inference task; before considering the
natural question this general connection prompts: ``Are the resulting
protocols optimal?''

\subsection{Private-coin distributed inference via distributed simulation}\label{sec:sampling:applications}

Having a distributed simulation protocol at our disposal, a natural
protocol for distributed inference entails using distributed
simulation to generate independent samples from the underlying
distribution, as many as warranted by the sample complexity of the
underlying problem, before running a sample inference algorithm (for
the centralized setting) at the referee. The resulting protocol will
require a number of players roughly equal to the sample complexity of
the inference problem when the samples are centralized times
$\big(\ab/2^\numbits)$, the number of players required to simulate
each independent sample at the referee. We refer to such protocols
that first simulate samples from the underlying distribution and then
use a standard sample-optimal inference algorithm at the referee
as \emph{simulate-and-infer} protocols. Formally, we have the
following result.

\begin{theorem}\label{theo:exactsampling:implies:ub}
Let $\task$ be an inference problem for distributions over a domain of
size $\ab$ that is solvable using $\psi(\task, \ab)$ samples with
error probability at most 1/3. Then, the simulate-and-infer protocol
for $\task$ requires at most
$O\left(\psi(\task, \ab)\cdot \frac{\ab}{2^{\numbits}}\right)$
players, with each player sending at most $\numbits$ bits to the
referee and the overall error probability at most $2/5$.
\end{theorem}
\begin{proof}
  The reduction is quite straightforward, and works in the following
   steps \begin{enumerate} \item Partition the players into blocks of
   size $54\ab/2^{\numbits}$.  \item Run the distributed simulation
   protocol (\cref{theo:generate:sample:lbits:1sample}) on each
   block.  \item Run the centralized algorithm over the simulated
   samples.  \end{enumerate}
   From~\cref{theo:generate:sample:lbits:1sample}, we have a Las Vegas
   protocol for distributed simulation using $27\ab/2^{\numbits}$
   players in expectation. Thus, by Markov's inequality, each block in
   the above protocol simulates a sample with probability at least
   1/2.  If the number of samples simulated is larger than
   $\psi(\task, \ab)$, then the algorithm has error at most
   1/3. Denoting the number of blocks by $B$, the number of samples
   produced has expectation at least $B/2$, and variance at most
   $B/4$. By Chebychev's inequality, the probability that the number
   of samples simulated being less than $B/2 -\sqrt{B/4}\sqrt{15}$ is
   at most 1/15. If $B>4\psi(\task, \ab)+8$, then $B/2
   -\sqrt{B}\sqrt{15/4}>\psi(\task,\ab)$. Since $1/3+1/15=2/5$, the
   result follows from a union bound.
\end{proof}
As immediate corollaries of the result, we obtain distributed
inference protocols for distribution learning and uniformity testing.
Specifically, using the well-known result that $\bigTheta{\ab/\eps^2}$
samples are sufficient to learn a distribution over $[\ab]$ to within
a total variation distance $\eps$ with probability 2/3, we obtain:
\begin{corollary}\label{coro:exactsampling:implies:learning}
  Let $\numbits\in\{1,\ldots, \log\ab\}$. Then, there exists an
  $\numbits$-bit private-coin $(\ab, \eps, 3/5)$-learning protocol for
  $\bigO{\frac{\ab^{2}}{2^\numbits \eps^2}}$ players.
\end{corollary}

\noindent From the existence of uniformity testing algorithms using $O(\sqrt{\ab}/\eps^2)$ samples~\cite{Paninski:08,VV:17,DGPP:17}, we obtain:
\begin{corollary}\label{coro:exactsampling:implies:testing:uniformity}
 Let $\numbits\in\{1,\ldots, \log\ab\}$. Then, there exists an
 $\numbits$-bit private-coin $(\ab, \eps, 3/5)$-uniformity testing
 protocol for $\bigO{\frac{\ab^{3/2}}{2^\ell \eps^2}}$ players.
\end{corollary}

Interestingly, a byproduct of this ``simulate-and-infer'' connection
(and, more precisely, of~\cref{coro:exactsampling:implies:learning})
is that the $\alpha$-simulation protocol
from~\cref{theo:generate:sample:lbits:1sample} has optimal number of
players, up to constants.
\begin{corollary}\label{coro:optimality:lasvegas}
  Let $\numbits \in\{1,\ldots, \log\ab\}$, and $\alpha \in
  (0,1)$. Then, any $\numbits$-bit public-coin (possibly adaptive)
  $\alpha$-simulation protocol for $\ab$-ary distributions must have
  $\ns=\Omega(\ab/2^\numbits)$ players.
\end{corollary}
\begin{proof}
  Let $\pi$ be any $\numbits$-bit $\alpha$-simulation protocol with
  $\ns$ players; by~\cref{theo:exactsampling:implies:ub}, and
  analogously to~\cref{coro:exactsampling:implies:learning} we have
  that $\pi$ implies an $\numbits$-bit $(\ab, \eps, 1/3)$-learning
  protocol for $\ns' = \bigO{\ns\cdot {\ab}/{\eps^2}}$
  players.\footnote{Improving the probability of success from $3/5$ to
  $1/3$ can be achieved by standard arguments, with at most a constant
  factor blowup in the number of players.}{} (Moreover, the resulting
  protocol is adaptive, private-, pairwise-, or public-coin,
  respectively, whenever $\pi$ is.) However, as shown
  in~\cref{app:learning:lb} (\cref{theo:learning:lb}), any
  $\numbits$-bit public-coin (possibly adaptive) $(\ab, \eps,
  1/3)$-learning protocol must have
  $\bigOmega{\ab^2/(2^\numbits\eps^2)}$ players. It follows that $\ns$
  must satisfy $\ns \gtrsim \ab/2^\numbits$, as claimed.
\end{proof}

\begin{remark}
We note that the learning upper bound
of~\cref{coro:exactsampling:implies:learning} appears to be
established in~\cite{HMOW:18} as well (with however, to the best of
our knowledge, completely different techniques). The authors
of~\cite{HOW:18} also describe a distributed protocol for distribution
learning, but their criterion is the $\lp[2]$ distance instead of
total variation.\footnote{We note that, based on a preliminary version of our manuscript on arXiv, the $\lp[2]$ learning upper bound of~\cite{HOW:18} was updated to use a ``simulate-and-infer'' protocol as well.} 
Finally, our learning lower bound
(\cref{app:learning:lb}), invoked in the proof of~\cref{coro:optimality:lasvegas} above, is established by adapting a similar lower bound from~\cite{HOW:18} which again applies to learning in the $\lp[2]$ metric.
\end{remark}

\subsection{Is distributed simulation essential for distributed inference?}\label{sec:conjecture}
In the previous subsection, we saw that it is easy to derive distributed learning and testing protocols from distributed sampling. However, the optimality of simulate-and-infer for uniformity testing using private-coin protocols is unclear. In fact, a natural question arises: \emph{Is the simulate-and-infer approach always optimal?} Note that such an optimality would have appealing implementation consequences, where one need not worry about the target inference application when designing the communication protocol -- the communication protocol will simply enable distributed simulation and the referee can implement the specific inference algorithm needed. A similar result, known as \emph{Shannon's source-channel separation theorem}, has for instance allowed for development of compression algorithms separately from the error-correcting codes for noisy channels. Unfortunately, optimality of simulate-and-infer can be refuted by the following simple example:
\begin{observation}
In the distributed setting model ($\ns$ players, and $\numbits=1$ bit of communication per player to the referee), testing whether a distribution over $[\ab]$, promised to be monotone, is uniform vs. $\eps$-far from uniform can be done with $\ns=O(1/\eps^2)$ (moreover, this is optimal).
\end{observation}
\begin{proof}[Sketch]
Optimality is trivial, since that many samples are required in the non-distributed setting. To see why this is enough, recall that a monotone distribution $\p\in\distribs{[\ab]}$ is \eps-far from uniform if, and only if, $\p(\{1,\dots,\ab/2\}) > \p(\{\ab/2+1,\dots,\ab\}) + 2\eps$. Therefore, we only need $\ns=O(1/\eps^2)$ players, where each player sends 1 if their sample was in $\{1,\dots,\ab/2\}$ and $0$ otherwise.
\end{proof}
However, this counter-example is perhaps not satisfactory since the inference problem itself was compressible since the dimension of the parameter space was increased artificially.\footnote{That is, the inference task was only ``superficially'' parameterized by $\ab$, but was actually a task on $\{0,1\}$ and entails only estimating the bias of a coin in disguise.}{} In fact, it is natural to consider an extension of simulate-and-infer where we first compress the observation to capture the effective dimension of the underlying parameter space. To formally define such \emph{compressed simulate-and-infer} schemes, we must first define a counterpart of sufficient statistic that will be relevant here. 

Let $\cP$ denote an inference problem (for instance, the distribution learning and the uniformity testing problem of the previous section) and $\ns(\cP)$  denote the minimum number of independent samples required to solve it, namely its sample complexity. Note that the description of the problem $\cP$ includes the observation alphabet $\cX$, the loss function used to evaluate the performance, and the required performance. For a (fixed, deterministic) mapping $f$ on $\cX$, denote by $\cP_f$ the problem where we replace each observed sample $X$ with $f(X)$. 

\begin{definition}[The size of a problem]
A problem $\cP$ is said to be \emph{compressible to $\numbits$ bits} if there exists a mapping $f\colon\cX\to \{0,1\}^\numbits$ such that $\ns(\cP_f) = n(\cP)$. For such a function $f$ with $\numbits \leq \log \abs{\cX}$, we call $f(X)$ a \emph{compressed statistic}. 

The \emph{size} $\abs{\cP}$ of a problem $\cP$ is then defined as the least $\numbits$ such that $\cP$ is compressible to $\numbits$. If a mapping $f$ attains $\abs{\cP}$, we call $f$ a \emph{maximally compressed statistic} for $\cP$. 
\end{definition}  

\begin{example}
For the uniformity testing problem $\UTesting(\ab,\eps)$ considered in the previous section, we must have
\[
  \abs{\UTesting(\ab,\eps)} \geq \log \ab - \log \frac{1}{1-\eps}.
\]
The proof follows by noting that for each mapping $f$ with range cardinality $\ab(1-\eps)$, we can find a distribution $Q$ on $[\ab]$ that is $\eps$-far from uniform, yet $f(X)$ is uniform under $Q$.
\end{example}
A compressed simulate-and-infer scheme then proceeds by replacing the original observation $X_j$ by its maximally compressed sufficient statistic
$f(X_j)$ at player $j$ and then applying simulate-and-infer for $f(X)$. Note that this new scheme, too, has the appealing feature that
we can use our distributed simulation protocol as a black-box communication step to enable distributed inference. But are such
compressed simulate-and-infer schemes optimal? Formally,
\begin{question}[The Flying Pony Question]\label{conj:flying:pony}
To solve $\cP$ in the distributed setting, must the number of parties $\ns$ satisfy $\ns= \bigOmega{ 2^{\abs{\cP}-\ell }\cdot n(\cP) }$?
\end{question}
This essentially asks whether the most economical communication scheme to solve $\cP$ is indeed to simulate $\ns(\cP)$ samples from a maximally compressed statistic. Observe that we noted the optimality of such a scheme for distribution learning, even when public-coin protocols are allowed. Further, to the
best of our knowledge, previous results on this topic, in essence, establish lower bounds to show the optimality of such simple schemes (cf.~\cite{GMNW:16,DGLNOS:17,HOW:18}.\smallskip

In spite of this evidence, we are able to refute this conjecture.\footnote{Thus implying that, even if wishes \emph{were} horses, there would be no flying ponies.}{} Specifically, we exhibit an inference task $\task$ over $\ab$-ary distributions which admits a $1$-bit private-coin protocol with $\ns=o(2^{\abs{\cP}} n(\cP))$ players.
\begin{theorem}\label{theo:refuting:fpc}
  There is an inference task $\task$ over $\ab$-ary distributions with $2^{\abs{\cP} }\cdot n(\cP) = \Omega({\ab^{3/2}})$, yet for which there exists a $1$-bit private-coin protocol with $\ns = \bigO{ \ab }$ players.
\end{theorem} 
\begin{proof}
We start by describing the inference task in question. For every even $\ab\geq 2$, $\task$ consists in distinguishing between the following two cases: either $\p=\uniformOn{\ab}$, the uniform distribution over $[\ab]$; or  $\p$ is any of the $2^{\ab/2}$ possible uniform distributions over a subset of size $\ab/2$ defined as follows. For a parameter $\theta\in\{-1,1\}^{\ab/2}$, $\p_{\theta}$ is the distribution such that, for every $i\in[\ab/2]$,
\[
    \p_{\theta}(2i-1) = \frac{1+\theta_i}{\ab}, \qquad \p_{\theta}(2i) = \frac{1-\theta_i}{\ab}
\]
and in particular $\totalvardist{\p_{\theta}}{\uniformOn{\ab}} = 1/2$ for every $\theta\in\{-1,1\}^{\ab/2}$. 

By an easy birthday-paradox type argument, we have that $\ns(\task) = \Omega(\sqrt{\ab})$ (and this is tight), so to prove the first part of the statement it is enough to show that $\abs{\task} = \Omega(\ab)$. To see why this is the case, set $L\eqdef \abs{\task}$, and consider any maximally compressed statistic $f\colon [\ab] \to \{0,1\}^{L}$ for $\task$. This $f$ immediately implies a (private-coin) $L$-bit $(\ab,1/2)$-uniformity testing protocol: namely, a protocol where each player first applies $f$ to their sample, then sends the resulting $L$ bits to the referee. Further, by definition of a maximally compressed statistic, we have $\ns(\task_f)=\ns(\task) = \Theta(\sqrt{\ab})$; as in the aforementioned $L$-bit protocol the referee only needs $\ns(\task_f)$ samples from the distribution on $2^L$, this therefore gives an $O(\sqrt{\ab}\cdot 2^{L-L}) = O(\sqrt{\ab})$ upper bound on the number of players required.

 However, peeking ahead,~\cref{theo:uniformity:shared:randomness} shows that any $L$-bit protocol for $\task$ (even allowing for public coins) must have $\Omega(\ab/2^{L/2})$ players.\footnote{Indeed, this is because the inference task $\task$ described here is a specific case of the lower bound construction underlying the proof of~\cref{theo:uniformity:shared:randomness}, obtained by taking $\eps=1/2$.}{} Combining this lower bound with the $O(\sqrt{\ab})$ upper bound we have just established yields
\begin{equation}\label{eq:refuting:fpc}
    \frac{\ab}{2^{L/2}}\lesssim \sqrt{\ab}\,,
\end{equation}
i.e., $\ab \lesssim 2^L$. This, along with the lower bound on $\ns(\task)$, implies that $2^{\abs{\cP} }\cdot n(\cP) = \Omega( \ab \cdot \sqrt{\ab}) = \Omega(\ab^{3/2})$, as claimed.\smallskip

To obtain a contradiction, it remains to prove the second part of the statement, i.e., to describe a $1$-bit private-coin protocol with $\ns=O(\ab)$ players. Consider the protocol where every of the $\ns$ players simply sends $1$ if their sample is equal to $1$, and $0$ otherwise. If $\p=\uniformOn{\ab}$, then each bit is independently $1$ with probability $1/\ab$. However, if $\p$ is one of the distributions uniform over $\ab/2$ elements, then $\p_1\in\{0,2/\ab\}$, and therefore either each player's bit is independently $1$ with probability $0$, or each player's bit is independently $1$ with probability $2/\ab$. In either case, the problem then amounts to distinguish a coin with bias $1/\ab$ to one with bias either $0$ or $2/\ab$; for which $\ns=O(\ab)$ players suffice, concluding the proof.
\end{proof}

While we have refuted the optimality compressed simulate-and-infer,
the strategy used in the counter-example above still entails 
simulating samples from a fixed distribution at the referee. This
statistic, while compressed form of the original problem, is not
a compressed sufficient statistic as it mandates a higher number of
samples in the centralized setting. We call such inference protocols
that entail simulate-and-infer for some compressed statistic of the
problem\footnote{This need not be a compressed sufficient statistic.}
\emph{generalized simulate-and-infer}; the optimality of generalized
simulate-and-infer is unclear, in general. For our foregoing example
of uniformity testing, it is not even clear whether there is a private-coin
protocol that requires fewer players than the vanilla
simulate-and-infer scheme. Interestingly, we can provide a public-coin
protocol that outperforms simulate-and-infer for uniformity testing
and show that it is optimal. This is the content of the next section.

\section{Public-Coin Uniformity Testing}\label{sec:uniformity}
In this section, we consider public-coin protocols for
 $(\ab, \eps)$-uniformity testing and establish the following upper
 and lower bounds for the required number of players.
\begin{theorem}\label{theo:uniformity:shared:randomness:ub}
For $1\leq \numbits \leq \log\ab$, there exists an $\numbits$-bit
public-coin $(\ab, \eps)$-uniformity testing protocol for
$\ns=\bigO{\frac{\ab}{2^{\numbits/2}\eps^2}}$ players.
\end{theorem}
Note that this is much fewer than the
$O(\ab^{3/2}/(2^\numbits\eps^2))$ players required using private-coin
protocols in~\cref{coro:exactsampling:implies:testing:uniformity}. In
fact, this is optimal, being the least number of players (up to
constant factors) needed for any public-coin protocol:
\begin{theorem}\label{theo:uniformity:shared:randomness:lb}
For $1\leq \numbits \leq \log\ab$, any $\numbits$-bit public-coin
$(\ab, \eps)$-uniformity testing protocol must have
$\ns=\bigOmega{\frac{\ab}{2^{\numbits/2}\eps^2}}$ players.
\end{theorem}

We establish~\cref{theo:uniformity:shared:randomness:ub}
and~\cref{theo:uniformity:shared:randomness:lb}
in~\cref{sec:uniformity:ub,sec:uniformity:lb}, respectively. Before
delving into the proofs, we note that the results for uniformity
testing imply similar upper and lower bounds for the more general
question of \emph{identity testing}, where the goal is to test whether
the unknown distribution $\p$ is equal to (versus $\eps$-far from) a
reference distribution $\q$ known to all the players.
\begin{corollary}\label{theo:identity:shared:randomness}
For $1\leq \numbits \leq \log\ab$, and for any fixed $\q\in\distribs{[\ab]}$, there exists an $\numbits$-bit
public-coin $(\ab, \eps,\q)$-identity testing protocol for
$\ns=\bigO{\frac{\ab}{2^{\numbits/2}\eps^2}}$ players. 
Further, any $\numbits$-bit public-coin
$(\ab, \eps,\q)$-identity testing protocol must have
$\bigOmega{\frac{\ab}{2^{\numbits/2}\eps^2}}$ players (in the worst case over $\q$).
\end{corollary}

We describe this
reduction (similar to that in the non-distributed setting)
in~\cref{app:identity:from:uniformity}, further detailing how it
actually leads to the stronger notion of ``instance-optimal'' identity testing in the sense
of Valiant and Valiant~\cite{VV:17}.

\subsection{Upper bound: public-coin protocols}\label{sec:uniformity:ub}
This section is dedicated to the proof
of~\cref{theo:uniformity:shared:randomness:ub}. We actually provide
and analyze two different protocols achieving the stated upper bound:
the first, in~\cref{sec:uniformity:ub:smooth}, is remarkably simple,
and, moreover, is ``smooth'' -- that is, no player's output depends
too much on any particular symbol from $[\ab]$. However, this first
protocol has the inconvenience of requiring a significant amount of
public randomness, $\Theta(\ab\cdot\numbits) = \Omega(\ab)$ bits.

To address this, we provide in~\cref{sec:uniformity:ub:levin} a
different protocol requiring the optimal number of players, too, but
necessitating much less randomness, only
$\Theta_\eps(2^\numbits\log\ab)=O_{\eps,\numbits}(\log\ab)$
bits.\footnote{Note that $2^\numbits\log \ab \leq \ab\numbits$ for
every $1\leq\numbits\leq\ab$.} On the other hand, this second protocol
is slightly more complex and highly ``non-smooth'' (specifically, the
output of each player entirely depends on only $\numbits$ symbols).

\subsubsection{A simple ``smooth'' protocol}\label{sec:uniformity:ub:smooth}

The protocol will rely on a generalization of the following
observation: \emph{if $\p$ is $\eps$-far from uniform, then for a
subset $S\subseteq[\ab]$ of size $\frac{\ab}{2}$ generated uniformly
at random, we have $\p(S) = \frac{1}{2}\pm \Omega(\eps/\sqrt{\ab})$,
with constant probability.} Of course, if $\p$ is uniform, then
$\p(S)=\frac{1}{2}$ with probability one. Further, note that this fact
is qualitatively tight: for the specific case of $\p$ assigning
probability $(1\pm\eps)/\ab$ to each element, the bias obtained will
be $\frac{1}{2}\pm \Theta(\eps/\sqrt{\ab})$ with high probability.

As a warm-up, we observe that the above claim immediately suggests a
protocol for the case $\numbits=1$: The $\ns$ players, using their
shared randomness, agree on a uniformly random subset
$S\subseteq[\ab]$ of size $\ab/2$, and send to the referee the bit
indicating whether their sample fell into this set. Indeed, if $\p$ is
$\eps$-far from uniform, with constant probability all corresponding
bits will be $(\eps/\sqrt{\ab})$-biased, and in this case the referee
can detect it with $\ns=O(\ab/\eps^2)$ players.\footnote{To handle the
small constant probability, it suffices to repeat this independently
constantly many times, on disjoint sets of $O(\ab/\eps^2)$ players.}

The claim in question, although very natural, is already non trivial
to establish due to the dependencies between the different elements
randomly assigned to the set $S$. We refer the reader
to~\cite[Corollary 15]{ACFT:18} for a proof involving
anticoncentration of a suitable random variable,
$Z\eqdef \sum_{i\in[\ab]} (\p_i-1/\ab)X_i$, with $X_1,\dots,X_\ab$
being (correlated) Bernoulli random variables summing to $\ab/2$. At a
high-level, the argument goes by analyzing the second and fourth
moments of $Z$, and applying the Paley--Zygmund inequality.

For our purposes, we need to show a generalization of the
aforementioned claim, considering balanced partitions into $L\eqdef
2^\numbits$ pieces instead of $2$. To do so, we first set up some
notation.  Let $L<\ab$ be an integer; for simplicity and with little
loss of generality, assume that $L$ divides $\ab$.  Further, with
$Y_1,\dots, Y_\ab$ independent and uniform random variables on $[L]$,
let random variables $X_1, \dots, X_\ab$ have the same distribution as
$Y_1, \dots, Y_\ab$ conditioned on the event that for every $r\in[L]$,
$\sum_{i=1}^\ab \indic{Y_i=r} = \frac{\ab}{L}$. Note that each $X_i$,
too, is uniform on $[L]$, but $X_i$s are not independent. For
$\p\in\distribs{[\ab]}$, define random variables $Z_1,\dots,Z_L$ as
follows:
\begin{equation} 
    Z_r \eqdef \sum_{i=1}^\ab \p_i \indic{X_i=r}\,.
\end{equation} 
Equivalently, $(Z_1,\dots, Z_L)$ correspond to the probabilities
$(\p(S_1),\dots, \p(S_L))$ where $S_1,\dots,S_L$ is a uniformly random
partition of $[\ab]$ into $L$ sets of equal size.

\begin{theorem}\label{theo:uniformity:z:concentration:anticoncentration:general}
  For the (random) distribution $\q=(Z_1,\dots, Z_L)$ over $[L]$
  induced by $(Z_1,\dots,Z_L)$ above, the following holds: (i) if
  $\p=\uniform$, then $\normtwo{\q-\uniform_L} = 0$ with probability
  one; and (ii) if $\lp[1](\p,\uniform) > \eps$,
  then \[ \probaOf{ \normtwo{\q-\uniform_L}^2 > \frac{\eps^2}{\ab}
  } \geq c\,.  \] for some absolute constant $c>0$.
\end{theorem}
The proof of this theorem is quite technical and is deferred
to~\cref{app:smooth}. We now explain how it yields a protocol
with the desired guarantees (i.e., matching the bounds
of~\cref{theo:uniformity:shared:randomness:ub}). By~\cref{theo:uniformity:z:concentration:anticoncentration:general},
setting $L=2^{\numbits}$ we get that with constant probability the
induced distribution $\q$ on $[L]$ is either uniform (if $\p$ was), or
at $\lp[2]$ distance at least $\eps'$ from uniform, where
$\eps' \eqdef \sqrt{\eps^2/\ab}$. However, testing uniformity
vs. $(\gamma/\sqrt{L})$-farness from uniformity in $\lp[2]$ distance,
over $[L]$, has sample complexity $O(\sqrt{L}/{\gamma}^2)$ (see
e.g.~\cite[Proposition 3.1]{CDVV:14} or~\cite[Theorem
2.10]{CDGR:17:journal}), and for our choice of
$\gamma\eqdef \sqrt{L}\eps'\in(0,1)$, we have
\begin{equation}\label{eq:l2:testing}
    \frac{\sqrt{L}}{\gamma^2} = \frac{\sqrt{L}}{L{\eps'}^2}
    = \frac{\ab}{\sqrt{L}\eps^2}= \frac{\ab}{2^{\numbits/2}\eps^2},
\end{equation} 
giving the bound we sought. This is the idea underlying the following
result:
\begin{corollary}\label{theo:uniformity:shared:randomness:ub:smooth}
For $1\leq \numbits\leq \log\ab$, there exists an
$\numbits$-bit public-coin $(\ab, \eps)$-uniformity testing protocol
for $\ns=\bigO{\frac{\ab}{2^{\numbits/2}\eps^2}}$ players, which uses
$O(\numbits\ab)$ bits of randomness.
\end{corollary} 
\begin{proof}
The protocol proceeds as follows: Let $m=\Theta(1)$ be an integer such
that $(1-c)^m \leq 1/6$, where $c$ is the constant
from~\cref{theo:uniformity:z:concentration:anticoncentration:general};
define $\delta \eqdef 1/(6m)$. Let $N
= \Theta(\ab/(2^{\numbits/2}\eps^2))$ be the number of samples
sufficient to test $(\eps/\sqrt{\ab})$-farness in
$\lp[2]$ distance from the uniform distribution over $[L]$, with
failure probability $\delta$ (as guaranteed
by~\eqref{eq:l2:testing}). Finally, let $\ns \eqdef mN 
= \Theta(\ab/(2^{\numbits/2}\eps^2))$. Given $\ns$ players, the
protocol divides them into $m$ disjoint batches of $N$ players, and each
group acts independently as follows:    
\begin{itemize}
  \item Using their shared randomness, the players choose uniformly at
  random a partition $\Pi$ of $[\ab]$ into subsets of size
  $\ab/2^\numbits$.
  \item Next, they send to the referee the $\numbits$ bits indicating which
  part of the partition their observed sample fell in. 
\end{itemize}
The referee, receiving these $N$ messages (which correspond to $N$
independent samples of the distribution $\q\in\distribs{[2^\numbits]}$
induced by $\p$ on $\Pi$) runs the $\lp[2]$ uniformity test, with
failure probability $\delta$ and distance parameter
$\eps/\sqrt{\ab}$. After running these $m$ tests, the referee rejects
if any of the batch is rejected, and accepts otherwise. 

By a union bound, all these $m$ tests will be correct with probability
at least $1-m\delta = 5/6$. If $\p=\uniform_\ab$, then all $m$ batches
generate samples from the uniform distribution on $[L]$, and the
referee returns \accept with probability at least $5/6$. However, if
$\p$ is $\eps$-far from uniform then with probability at least
$1-(1-c)^m \geq 5/6$ at least one of the $m$ groups will choose a
partition such that the corresponding induced distribution on $[L]$ is
at $\lp[2]$ distance at least $\eps/\sqrt{\ab}$ from uniform; by a
union bound, this implies the referee will return \reject with
probability at least $1-2\cdot1/6 = 2/3$. 

The bound on the total amount of randomness required comes from the
fact that $m=\Theta(1)$ independent partitions of $[\ab]$ into
$L\eqdef 2^\numbits$ are chosen and each such partition can be specified
using $O(\log(L^\ab)) = O(\ab\cdot \numbits)$ bits. 
\end{proof}
Note that the protocol underlying~\cref{theo:uniformity:shared:randomness:ub:smooth} is ``smooth,'' in the sense that each
player's output is the indicator of a set of $\ab/2^\numbits$
elements, which for constant values of $\numbits$ is $\Omega(\ab)$.
\subsubsection{A randomness-efficient optimal protocol}\label{sec:uniformity:ub:levin}
We now provide our second optimal public-coin protocol, which albeit less simple than that of the previous section is much more randomness-efficient. We start with the case
$\numbits=1$ (addressed in \cref{prop:uniformity:shared:randomness}
below), before generalizing to an arbitrary $\numbits\geq 1$ -- the generalization is nontrivial and uses a more involved protocol. Before we present our actual scheme, to help the reader build heuristics, we present a simple, albeit non-optimal, scheme. 
\begin{proposition}[Warmup]\label{prop:uniformity:shared:randomness:warmup}
There exists a $1$-bit public-coin $(\ab, \eps)$-uniformity testing protocol for $\ns=\bigO{\ab\log(1/\eps)/\eps^3}$, which uses $\bigO{(\log \ab)/\eps}$ bits of randomness.
\end{proposition}
\begin{proof}
The starting point of the protocol is the straightforward observation that if $\p$ is $\eps$-far from uniform then at least an $\bigOmega{\eps}$ fraction of the domain must have an $\bigOmega{\eps}/\ab$ deviation from uniform. Indeed, consider a $\p\in\distribs{[\ab]}$ such that $\totalvardist{\p}{\uniform_{[\ab]}}\ge\eps$. By contradiction, suppose that there are only $\ab' < \frac{\eps}{2}\cdot \ab$ elements such that  $\p_i< (1-\frac \eps 2)\cdot \frac 1\ab$. Then, 
\[
  \totalvardist{\p}{\uniform_{[\ab]}}
  = \sum_{i : \p_i < 1/\ab} \left(\frac{1}{\ab} - \p_i\right)
  \leq \ab'\cdot \frac{1}{\ab} + (\ab - \ab')\cdot \frac{\eps}{2\ab}
  < \frac{\eps}{2} + \frac{\eps}{2} = \eps
\]
contradicting the assumption that $\p$ was $\eps$-far from uniform. Therefore, 
\begin{align}\label{e:desired_set}
|\big\{i\in [\ab]: \p_i < (1-\frac \eps 2)\cdot \frac 1 \ab \big\}|\geq \frac \eps 2\cdot \ab.
\end{align}
Next, we recall the well-known fact that a coin with bias $1/\ab$ can
 be distinguished from another with bias $(1-\eps/2)/\ab$ with
 probability\footnote{That is, denoting the two distributions by $P$
 and $Q$, we can find a subset $A$ of sequences in $\{0,1\}^{\ns}$ such that $P^{\ns}(A)\geq
 1-\delta$ and $Q^{\ns}(A)\leq \delta$.} $1-\delta$ using $c\ab \log
 (1/\delta)/\eps^2$ independent coin tosses for some constant $c$. Therefore, any $i\in[\ab]$ with $\p_i < (1-\eps/2)/\ab$ 
 were known to the players, then testing if its probability is $1/\ab$
 or $(1-\eps/2)/\ab$ will require $c\ab\log 1/\delta)/\eps^2$ players simply by using $\indic{X_j=i}$ as communication for player $j$. 

We use this observation to build our protocol. Specifically, we divide
the players into $m$ disjoint batches of size
$\ns'=c \ab\log(1/\delta)/\eps^2$ players; we will specify $m$ and
$\delta$ later. We assign 
a random element $i$ to each batch, generated uniformly from $[\ab]$
using public randomness. Then, the parties in the batch apply the
aforementioned test to distinguish if the probability of the selected
$i$ is $1/\ab$ or $(1-\eps/2)/\ab$. We accept $\uniform$
if all the batches accepted $1/\ab$ as the probability of their
respectively assigned $i$s; else we reject $\uniform$. Note that since
each batch's selected $i$ lies in the desired set in \eqref{e:desired_set}
with probability at least $\eps/2$, with probability greater than
$9/10$ at least one batch will be assigned an $i$ in the desired set
if the number of batches satisfies
\[
m\geq \frac 5\eps.
\]
When the underlying
distribution is uniform, the  protocol will make an error only if one
of the test for one of the batches fails, which can happen with
probability less than $m\delta\leq 1/10$ if 
\[
\delta \leq 1/(10m). 
\]
On the
other hand, if the underlying
distribution is $\eps$-far from uniform, then the test will fail only
if either no $i$ in the desired set was selected or if the protocol
failed for an $i$ in the desired set; the former happens with
probability less than $1/10$ and the latter with probability less than
$1/10m$ when we choose $\delta=1/10m$. Thus, the overall probability
of error in this case is less than $2/10$, whereby we have an
$(\ab, \eps)$-uniformity testing protocol using  $m\ns'=5c\ab\log
(50/\eps)/\eps^3$ players. Moreover, the total number of random bits required is $O(m\log\ab)$, as claimed.
\end{proof} 

\paragraph{Improving on the warmup protocol using Levin's work investment strategy.} The main issue with the proof
of~\cref{prop:uniformity:shared:randomness:warmup} is the use of a
``reverse'' Markov style argument to identify the $i$ to focus
on. This approach is 
inherently wasteful, as can be seen by considering the two extremes
cases of distribution $\p$ such that $\totalvardist{\p}{\uniform}
> \eps$: First, when a constant fraction of the elements have
probability $(1-\bigOmega{\eps})/\ab$ and the rest have probability
more than $1/\ab$, in which case we only need
$m=\bigO{1}$ batches to find such an element $i$ and
$\bigO{\ab/\eps^2}$ players per batch to detect the bias. Second, when
a fraction $\bigO{\eps}$  of the elements have probability $0$ and
the rest have probability more than $1/\ab$, in
which case we need $m=\bigO{1/\eps}$ batches to find such
an element, but now only $\bigO{\ab/\eps}$ players per batch to
detect the bias. In both cases, the total number of players should be
$\bigO{\ab/\eps^2}$, in contrast to the 
$\bigO{\ab\log 1/\eps/\eps^3}$ of the scheme described above. To
circumvent this difficulty, we take recourse to a technique known
as \emph{Levin's work investment strategy}; see, for
instance,~\cite[Appendix A.2]{Goldreich:14} for a review
(cf.~\cite[Section~8.2.4]{Goldreich:17}). Heuristically, this technique
allows us to identify an appropriate ``scale'' and invests matching
``work'' effort to it. Formally, we have following lemma:
\begin{lemma}[{\cite[Fact A.2]{Goldreich:14}}]\label{fact:uniformity:shared:randomness:levin}
Consider a random variable $X$ taking values in  $\cX$, a mapping
$q\colon \cX\to [0,1]$, and $\eps\in(0,1]$. Suppose
that $\shortexpect[ q(X) ] >\eps$, and let
$L\eqdef \clg{\log(2/\eps)}$. Then, there exists $j^\ast\in[L]$ such
that $\probaOf{ q(X) > 2^{-j^\ast} } >
2^{j^\ast}\eps/(L+5-j^\ast)^2$. 
\end{lemma}
The next result follows upon modifying the warmup protocol by using
this lemma to decide on the size and the number of batches.

\begin{proposition}\label{prop:uniformity:shared:randomness}
There exists a $1$-bit public-coin $(\ab, \eps)$-uniformity testing protocol for $\ns=\bigO{\ab/\eps^2}$, which uses $\tildeO{(\log \ab)/\eps}$ bits of randomness.
\end{proposition}
\begin{proof}
Consider a $\p$ that is \eps-far from uniform and set
$L\eqdef \clg{\log(2/\eps)}$. We
apply~\cref{fact:uniformity:shared:randomness:levin} to the
function $q\colon[\ab]\to[0,1]$ given by
\[
  q(i) \eqdef \ab\left(\frac{1}{\ab} - \p_i\right) \indic{\p_i < 1/\ab}.
\]
Note that for a uniformly distributed $X$, we have
\[
 \shortexpect[ q(X) ] = \sum_{i=1}^\ab \frac{1}{\ab} q(i) = \totalvardist{\p}{\uniform} > \eps.
\] 
Therefore, by~\cref{fact:uniformity:shared:randomness:levin}, there exists $j^\ast\in[L]$ such that 
\begin{equation}\label{e:good_j}
    \probaDistrOf{ i\sim \uniform }{ \p_i < \frac{1}{\ab} - \frac{2^{-j^\ast}}{\ab} } > 2^{j^\ast}\frac{\eps}{(L+5-j^\ast)^2}\,.
\end{equation}
We now proceed in a similar manner as the warmup protocol, with one
batch invested for each $j\in [L]$. Specifically, consider $L$
batches of players with the $j$-th batch comprising $\ns_j \eqdef
m_j\cdot \bigO{2^{2j}\ab\cdot \log(1/\delta_j) }$ players; both quantities
$\delta_j$ and $m_j$ will be specified later. The $j$th batch assumes
that $j\in[L]$ will satisfy \eqref{e:good_j} and further divides its
assigned set of players into $m_j$ mini-batches. Each mini-batch
selects an $i$ uniformly from $[\ab]$ and applies the previously
mentioned test to distinguish if the probability of $i$ is $1/\ab$
or $(1-2^{-j})/\ab$ with probability $(1-\delta_j)$. This, as before,
requires the assigned $\bigO{2^{2j}\ab\log (1/\delta_j)}$
players. Note that if 
\[
m_j\geq 5(L+5-j)^2/(2^j\eps),
\]
whenever $j$ satisfies \eqref{e:good_j}, with probability more than $9/10$
at least one of the mini-batches assigned to $j$ will select an $i$
for which $\p_i<(1-2^{-j})/\ab$. Our uniformity testing protocol is as
before: 

Accept uniform if every mini-batch of every batch accepts
$1/\ab$ as the probability of their respectively assigned elements;
else declare $\eps$-far from uniform.

If the underlying distribution is indeed uniform, the protocol will
reject it when at least one of the mini-batches erroneously rejects $1/\ab$, an
event which occurs with probability at most
\[
\sum_{j=1}^L m_j \delta_j\leq \frac 1 {10} \sum_{j=1}^L \frac1
{(L+5-j)^2} < \frac 1{10}\sum_{j\geq 5} \frac 1{j^2} < \frac 1{40},
\]
when we select
\[
\delta_j \leq \frac 1{10(L+5-j)^2m_j}.
\]
Note that this choice of $\delta_j$ depending on $j$ is important for
omitting the extra $\log(1/\eps)$ cost that appeared in the warmup 
protocol. 

If the underlying distribution is $\eps$-far from uniform, an error
occurs if, for every $j$, no mini-batch of batch $j$ selects an
element $i$ that satisfies \eqref{e:good_j} or if the mini-batch test
fails. By construction, there exists a $j$ that
satisfies \eqref{e:good_j}, and by our choice of $m_j$, all the
mini-batches assigned to it fails to select an $i$ in the set
of \eqref{e:good_j} with probability less than $1/10$. On the other
hand, the minibatch makes an error with probability less than
$\delta_j<1/40$. Thus, the overall probability of error in this case
is less than $1/8$. To conclude, the overall protocol makes an error
with probability at most $1/8$ in both cases.

Finally, we can bound the total number of players $\ns$ required by
the protocol as
\begin{align*}
\ns &= c\sum_{j=1}^L m_j \cdot 2^{2j}\ab \cdot \log \frac 1{\delta_j}
\\
&\leq \frac{c\ab}{\eps} \sum_{j=1}^L  (L+5-j)^2 2^{j} \log \frac
{10(L+5-j)^4}{2^j \eps} 
\\
&\leq \bigO{\frac{\ab}{\eps}} \sum_{j=1}^L  (L+5-j)^2 2^{j} \log \frac
{2^{j-L-5}} {L+5-j}
\\
&\leq \bigO{\frac{\ab}{\eps}}2^L\sum_{j'=5}^L  (j')^2 2^{-j'} \log (j'2^{j'})
\\
&\leq \bigO{\frac{\ab}{\eps}}2^L\sum_{j'=5}^L  (j')^2
2^{-j'} \log {j'}
+ \bigO{\frac{\ab}{\eps}}2^L\sum_{j'=5}^L  (j')^3 2^{-j'}
\\
&\leq \bigO{\frac {\ab}{\eps^2}},
\end{align*}
where the final bound uses $\sum_{j\geq 5} 2^{-j} j^\alpha \log j = \bigO{1}$ for every
$\alpha$. To conclude, note that the total number of random bits required is $O(\log \ab)\cdot \sum_{j=1}^L m_j = O(\log \ab\cdot\log^2(1/\eps)/\eps)$.
\end{proof}
Next, we move to the more general case when $\numbits\geq 1$ bits of
communication per player are allowed and
establish~\cref{theo:uniformity:shared:randomness:ub}. While we build
on the heuristics developed thus far, the form of the protocol
deviates. Instead of assigning one symbol to each mini-batch, we now
assign a subset of size $s=2^\numbits -1$ to each mini-batch; one $\numbits$-bit
sequence is reserved to indicate that none of the symbol in the subset
occured. The referee uses the symbols occuring in the
subset to distinguish uniform and $\eps$-far from
uniform, which can be done when conditional distributions (i.e., given that
the symbols lie in the subset) are separated in total variation
distance. We use Levin's work investment strategy again to decide how
many players must be assigned to each subset. But now there is a new
constraint: If a subset has small probability, then we need to assign
a large number of mini-batches to get symbols from it. However, we can
circumvent this difficulty by noting that if the subset has
probability smaller than $(1-\eps)s/\ab$, we can anyway distinguish the
underlying distribution from uniform using $\bigO{\ab/(s\eps^2)}$
players. Thus, we can condition on the complementary event by adding extra
$\bigO{\ab/(s\alpha^2)}$ players per batch and take
$\bigO{k/s}$ mini-batches to get at least one mini-batch assigned
to a good subset. Note that under the uniform distribution the conditional distribution
given a subset of size $s$ is uniform on $[s]$. Then, we can
distinguish the conditional distributions using roughly\footnote{The
good subset need not have the conditional distributions separated by
exactly $\eps$, but this is where we use Levin's work investment
strategy to get an effective separation of $\eps$.}
$\bigO{\sqrt{s}/\eps^2}$ players.  The overall number of players is
dominated by the players assigned for the conditional test and is
given by $\bigO{\ab/(2^{\numbits/2}\eps^2)}$. We provide the formal proof next. 

\begin{theorem}\label{theo:uniformity:shared:randomness:ub:levin}
Let $\numbits\in\{1,\ldots, \log\ab\}$. Then, there exists an $\numbits$-bit public-coin $(\ab, \eps)$-uniformity testing protocol for $\ns=\bigO{\frac{\ab}{2^{\numbits/2}\eps^2}}$ players, which uses $O_{\eps}(2^\numbits\log\ab)$ bits of randomness.
\end{theorem}
\begin{proofof}{\cref{theo:uniformity:shared:randomness:ub:levin}}
Set $L\eqdef \clg{\log(2/\eps)}$ and define $q$ as in the proof
of~\cref{prop:uniformity:shared:randomness}. For a subset
$S\subseteq[\ab]$, denote by $\p^S$ the conditional distribution
\[
\p^S_i = \frac{\p_i\indic{i\in S}}{\p(S)},
\]
where $\p(S) = \sum_{i\in S}\p_i$.  Observe that if $\p=\uniform$,
then for every subset $S\subseteq[\ab]$ we have $\p^S_i = 1/|S|$ for
every $i \in S$. On the other hand,  when $\p$ is $\eps$-far from
uniform, we have the following result:
\begin{claim}\label{claim:uniformity:l:bits:expected:restricted:tv}
    Suppose $\totalvardist{\p}{\uniform}>\eps$. For any $1\leq
    s\leq \ab$ and $S\subseteq[\ab]$ of size $s$ chosen uniformly at
    random, we have
    \[
        \shortexpect_S \sum_{i\in S} \indic{\p_i \leq \frac{1}{\ab} }\left( \frac{1}{\ab} - \p_i \right) > \eps\cdot\frac{s}{\ab}\,.
    \]
\end{claim}
\begin{proof}
On expanding the expectation, we obtain
\begin{align*}
  \shortexpect_S \sum_{i\in S} &\indic{\p_i \leq \frac{1}{\ab} }\left( \frac{1}{\ab} - \p_i \right)
\\
  &= \frac{1}{\binom{\ab}{s}} \sum_{S\subseteq [\ab]: \abs{S} = s } \sum_{i\in S} \indic{\p_i \leq \frac{1}{\ab} }\left( \frac{1}{\ab} - \p_i \right)  
    = \frac{1}{\binom{\ab}{s}} \sum_{i=1}^\ab \sum_{\substack{S\subseteq [\ab]: \abs{S} = s \\ S\ni i}} \indic{\p_i \leq \frac{1}{\ab} }\left( \frac{1}{\ab} - \p_i \right)  \\
    &= \frac{1}{\binom{\ab}{s}} \sum_{i=1}^\ab \indic{\p_i \leq \frac{1}{\ab} }\left( \frac{1}{\ab} - \p_i \right) \sum_{\substack{S'\subseteq [\ab]\setminus i \\ \abs{S'} = s-1}} 1  
    = \frac{1}{\binom{\ab}{s}} \sum_{i=1}^\ab \indic{\p_i \leq \frac{1}{\ab} }\left( \frac{1}{\ab} - \p_i \right) \cdot \binom{\ab-1}{s-1} \\
    &= \frac{s}{\ab} \sum_{i=1}^\ab \indic{\p_i \leq \frac{1}{\ab} }\left( \frac{1}{\ab} - \p_i \right) \\
    &= \frac{s}{\ab} \cdot \totalvardist{\p}{\uniform} > \eps\cdot\frac{s}{\ab}\,.
\end{align*}
\end{proof}
For brevity, set $s\eqdef 2^\numbits-1$.
Using~\cref{claim:uniformity:l:bits:expected:restricted:tv} together
with~\cref{fact:uniformity:shared:randomness:levin} we get that there exists $j^\ast\in[L]$ such that 
\begin{equation}
    \probaDistrOf{ S }{ \sum_{i\in S} \indic{\frac{\ab\p_i}{s} \leq \frac{1}{s} }\left( \frac{1}{s} - \p_i\cdot\frac{k}{s} \right) > 2^{-j^\ast} } > 2^{j^\ast}\cdot\frac{\eps}{(L+5-j^\ast)^2}. 
\label{e:Levin_set}
\end{equation}
Note that the event on the left-side of the inequality above
essentially bounds the total variation distance between $\uniform_S$
and $\p^S$ when $\p(S) \approx s/\ab$. Therefore, in case players have
access to such a subset $S$, they can accomplish uniformity testing by
applying a standard uniform test for a domain of size $s$. Thus, as in
the $\numbits=1$ protocol, we can use a public randomness to select a
subset $S$ randomly and assign it to an appropriate number of
players; this constitutes one mini-batch, and we need one mini-batch
per $j\in[L]$. However, this will only work if the selected $S$ has
$\p(S)\approx s/\ab$. To circumvent this difficulty, we use a separate
test for checking closeness of $\p(S)$ to $s/\ab$.  Specifically, we
once again use the fact that a coin with bias $s/\ab$ can
 be distinguished from another with bias outside the interval
 $[(1-\alpha)s/\ab, (1+\alpha)s/\ab$ with probability of error less
 than $\delta$ using $\bigO{\ab \log (1/\delta)/(s\alpha^2)}$
 independent coin tosses.

Once we have verified that the set $S$ has probability close to
$s/\ab$, we can apply a standard uniformity test. Indeed, we set
$\alpha=2^{-j^\ast}/8$, and once the test above has verified
that $\p(S) \in [1-2^{-j^\ast}/8,
1+2^{-j^\ast}/8]\cdot \frac{s}{\ab}$, the total variation distance of
$\p^S$ to uniformity can be bounded as follows: Let
vector $\tilde{\p}^S$ be given by $\tilde{\p}^S_i = \p_i \ab/s$. Then,
by the triangle inequality we get
\[
  \totalvardist{\p^S}{\uniform_S} = \frac{1}{2}\normone{\p^S-\uniform_S} \geq \frac{1}{2}\left( \normone{\tilde{\p}^S - \uniform_S } - \normone{\tilde{\p}^S - \p^S } \right)
  > \frac{1}{2}\left( 2^{-j^\ast} - \normone{\tilde{\p}^S - \p^S } \right).
\]
Further, 
\[
\normone{\tilde{\p}^S - \p^S } = \sum_{i\in S} \p_i \abs{\frac{\ab}{s}
- \frac{1}{\p(S)}} \leq \frac{2^{-j}\ab}{4s}\cdot \p(S) \leq \frac{2^{-j}}{3},
\]
which gives
\begin{equation}\label{eq:uniformity:l:bits:conditional:tv}
    \totalvardist{\p_S}{\uniform_S} \geq \frac{2^{-j}}3.
\end{equation}
Therefore, upon the first test verifying that $\p(S)$ is sufficiently
close to $s/\ab$, we can proceed to testing $\p^S$ versus
$\uniform_S$. To that end, we need sufficiently many samples from the $\p^S$, which we generate using rejection sampling.

We have now collected all the components needed for our scheme. 
As in the $\numbits=1$ case, set parameters $\eps_j=2^{-j}/8$, $m_j=(L+5-j)^2/(2^j \eps)$, and
$\delta_j= 1/(10(L+5-j)^2m_j)$. Consider $L$ batches of
players, with the $j$th batch comprising $m_j$ mini-batches 
of 
\[
n_j = c_1\left(\frac{\ab}{s\eps_j^2} \log \frac 1 {\delta_j}\right)+ c_2\left(\frac{\ab}{s}\log \frac 1
{\delta_j}\right)\cdot c_3\left(\frac{\sqrt{s}}{\eps_j^2} \log \frac 1 {\delta_j}\right)
\]
players each; the constants $c_1, c_2,c_3$ will be set to get
appropriate probability of errors. Each mini-batch of the $j$th batch generates a random
subset $S$ of $[\ab]$. Each player in the mini-batch communicates as
follows: It sends the all-zero sequence of length $\numbits$ to indicate
if its observed element is not in $S$ and, otherwise, uses the remaining sequences
of length $\numbits$ to indicate which of the $s=2^\numbits -1$ elements it
has observed. The referee uses the communication from the first
$(c_1\ab\log 1/\delta_j/(s\eps_j^2))$ players of the
mini-batch to check if the $|\p(S)-
s/\ab|< \eps_js/\ab$ or not. If it is not, the mini-batch fails. Else,
the referee considers the communication from players that did not send
the all-zero sequence (i.e., those players that saw elements in $S$)
and tests if the conditional distribution $\p^S$ is uniform on $S$ or
not. If it is not, the mini-batch fails; the referee declares
uniformity if none of the mini-batches declared failure.

The analysis of this protocol is completed in a similar manner to that
of the 
$\numbits=1$ protocol. 
If the underlying distribution is uniform, the
output is erroneous if at least one of the mini-batches declared
failure. This can happen in two ways: First, if the test based on communication
from the first set of $c_1\ab\log (1/\delta_j)/s\eps_j^2$ players
erroneously declared fail, an event that can happen with probability
$\delta_j/3$ for an appropriately chosen constant $c_1$. Second, if the first test passes, but the uniformity test for $\p^S$ based on the
remaining remaining set of players fails, which can happen either when
there are less than $c_2\sqrt{s}/\eps_j^2\log (1/\delta_j)$ players that see samples from $S$, which  given
that the first set has passed will fail with probability less than
 $\delta_j/3$ for an appropriate $c_2$, or when the second test fails which for an appropriate $c_3$ will happen with
 probability less than $\delta_j/3$ by~\cref{eq:uniformity:l:bits:conditional:tv}. Thus, the overall probability
 of error is less than $\sum_{j=1}^L m_j \delta_j< 1/40$.

If the underlying distribution is $\eps$-far from uniform, the test
will erroneously select the uniform if the referee makes an error for
the mini-batches corresponding to $j^*$ guaranteed
by \eqref{e:Levin_set}. But this can only happen if either none of
these mini-batches select an $S$ satisfying the condition on the
left-side of \eqref{e:Levin_set}, which happens with probability less
than $1/10$ or
a mini-batch that selected an appropriate $S$ failed the second test,
which can happen with probability $\delta_j/3 \leq 1/10$. Thus, the
overall probability of error is less than $2/10$. 

We complete the
proof by evaluating the number of players used by the protocol. As in
the proof for the $\numbits=1$ case, we have that the total number of
players $n$ satisfies
\begin{align*}
n &\leq \sum_{j=1}^L m_j\bigO{\frac{\ab}{s\eps_j^2}\log \frac 1
{\delta_j} +  \frac{\ab}{\sqrt{s}\eps_j^2}\left(\log \frac 1
{\delta_j}\right)^2}
\\
&\leq \sum_{j=1}^L m_j\bigO{\frac{\ab}{\sqrt{s}\eps_j^2}\left(\log \frac 1
{\delta_j}\right)^2}
\\
&\leq \bigO{\frac \ab {\sqrt{s}\eps}} 2^L\sum_{j=5}^L 2^{-j}(j)^5(\log j)^2
\\
&\leq \bigO{\frac{\ab}{2^{\frac{\numbits}{2}}\eps^2}},
\end{align*}
where we followed the same steps as the bound for $\numbits=1$ case. To conclude, note that since each subset $S$ of size $s$ requires $\log\binom{\ab}{s} = O(2^\numbits\log\frac{\ab}{2^\numbits})$ bits to specify, the total number of random bits required is $O(2^\numbits\log\frac{\ab}{2^\numbits})\cdot \sum_{j=1}^L m_j = O(2^\numbits\log\frac{\ab}{2^\numbits}\cdot\log^2(1/\eps)/\eps)$.
\end{proofof}
 \subsection{Lower bound for public-coin protocols}\label{sec:uniformity:lb}
We now establish a lower bound on the number of players $\ns$ required
for any $\ell$-bit public-coin $(\ab,\eps)$-uniformity testing
protocol. We begin with the simpler setting of $\ell=1$ and
establish~\cref{theo:uniformity:shared:randomness:lb} for this special
case, before moving to the more general case. The same construction is
used for both the restricted and the general case, but the simpler
proof we present for the special case does not yield the more general
result.
\begin{proposition}\label{prop:uniformity:shared:randomness:lb}
  Any $1$-bit public-coin $(\ab, \eps)$-uniformity testing protocol
  must have $\ns=\bigOmega{\ab/\eps^2}$ players.
\end{proposition}
\begin{proof}
Without loss of generality, we assume that $k$ is even. We consider
the standard ``Paninski construction'':\footnote{This construction was
  given in~\cite{Paninski:08} to prove the lower bound for the sample
  complexity of uniformity testing in the standard centralized
  setting.}{} For every $\theta\in\{-1,1\}^{\ab/2}$, let
\[
    \p_\theta(2i-1) = \frac{1+2\eps \theta_i}{k}, \qquad
    \p_\theta(2i-1) = \frac{1-2\eps \theta_i}{k}, \quad \forall\,
    i\in[\ab/2].
\]
Note that each distribution $\p_\theta$ is at a total variation
distance $\eps$ from the uniform distribution $\uniform$ on
$[\ab]$. Thus, any $(\ab, \eps)$-uniformity testing protocol should be
able to distinguish between any distribution $\p_\theta$ and
$\uniform$. We will establish an upper bound on the average total
variation distance between $\p_\theta$ and $\uniform$, whereby there
must be at least one $\p_\theta$ satisfying the bound. The desired
lower bound for the number of players will then follow from the
standard Le Cam's two-point method argument.

We derive the aforementioned upper bound on average total variation
distance for private-coin protocols first. Specifically, consider a
private-coin protocol for uniformity testing where, as before, the
$1$-bit communication of player $j$ is described by the channel
$W_j\colon [\ab]\to \{0,1\}$ such that $W_j(1|x)\in[0,1]$ is the
probability that player $j$ sends $1$ to the referee upon observing
$x$. For any $1\leq j\leq \ns$, it is immediate to see that, if
$\uniform$ is the underlying distribution, the probability that player
$j$ sends $1$ to the referee is
\[
    \rho_{j}^{\uniform} \eqdef \frac{2}{\ab}\sum_{i=1}^{\ab/2} \left(
    \frac{W_j(1|2i-1)+W_j(1|2i)}{2} \right),
\]
while under $\p_\theta$ it is
\begin{align*}
    \rho_{j}^{\theta} &\eqdef \frac{2}{\ab}\sum_{i=1}^{\ab/2} \left(
    \frac{W_j(1|2i-1)+W_j(1|2i) +
      2\eps\theta_i\left(W_j(1|2i-1)-W_j(1|2i)\right)}{2} \right)
    \\ &= \rho_{j}^{\uniform} +
    \frac{\eps}{\ab}\sum_{i=1}^{\ab/2}\theta_i\left(
    W_j(1|2i-1)-W_j(1|2i) \right)\,.
\end{align*}
Moreover, since each player gets an independent sample from the same
distribution
$\p\in\{\uniform\}\cup\{\p_\theta\}_{\theta\in\{-1,1\}^{\ab/2}}$, the
observation of the referee $r\in\{0,1\}^{\ns}$ is generated from a
product distribution. Specifically, the bits communicated by the
players are independent with the $j$th bit distributed as
$\bernoulli{\p_j^\uniform}$ or $\bernoulli{\p_j^\theta}$,
respectively, when the underlying distribution of the sample are
$\uniform$ or $\p_\theta$. Denoting by $\mathbf{R}^\uniform$ and
$\mathbf{R}^\theta$ the distributions of the transmitted bits under
$\uniform$ and $\p_\theta$, repectively, we have
\begin{align*}
      \totalvardist{\mathbf{R}^{\uniform}}{\mathbf{R}^{\theta}}^2
      &\leq
      \frac{1}{2}D\left(\mathbf{R}^{\theta}\|\mathbf{R}^{\uniform}
      \right) \\ &=
      \frac{1}{2}\sum_{j=1}^{\ns}D\left(\bernoulli{\p_j^\theta}\|\bernoulli{\p_j^\uniform}\right)
      \\ &\leq\frac{1}{2}\sum_{j=1}^{\ns}\chi^2(\bernoulli{\p_j^\theta},\bernoulli{\p_j^\uniform})
      \\ &=\frac{1}{2}\sum_{j=1}^{\ns}\frac{ (\rho_{j}^{\uniform} -
        \rho_{j}^{\theta}
        )^2}{\rho_{j}^{\uniform}(1-\rho_{j}^{\uniform})},
\end{align*}
where the first inequality is Pinsker's inequality, the second
inequality uses $\ln x\leq (x-1)$. Further, abbreviating for
convenience $\alpha_{i,j} \eqdef W_j(1|2i-1)$ and $\beta_{i,j} \eqdef
W_j(1|2i)$ for $i\in[\ab/2]$ and $j\in[\ns]$, to bound the right-side
we note that
\begin{align*}
 \frac{ (\rho_{j}^{\uniform} - \rho_{j}^{\theta}
   )^2}{\rho_{j}^{\uniform}(1-\rho_{j}^{\uniform})} &=
 \frac{4\eps^2}{\ab^2}
 \frac{1}{\rho_{j}^{\uniform}(1-\rho_{j}^{\uniform})} \left(
 \sum_{i=1}^{\ab/2} \theta_i\left( \alpha_{i,j}-\beta_{i,j} \right)
 \right)^2 \\ &= \frac{4\eps^2}{\ab^2}
 \frac{1}{\rho_{j}^{\uniform}(1-\rho_{j}^{\uniform})} \sum_{i,
   i'=1}^{\ab/2} \theta_i\theta_i'\left( \alpha_{i,j}-\beta_{i,j}
 \right)\left( \alpha_{i',j}-\beta_{i',j} \right).
\end{align*}
On taking expectation over $\theta$, we obtain
\begin{align*}
\bE{\theta}{\frac{ (\rho_{j}^{\uniform} - \rho_{j}^{\theta}
    )^2}{\rho_{j}^{\uniform}(1-\rho_{j}^{\uniform})} } &=
\frac{4\eps^2}{\ab^2}
\frac{1}{\rho_{j}^{\uniform}(1-\rho_{j}^{\uniform})} \sum_{i,
  i'=1}^{\ab/2} \bE{\theta}{\theta_i\theta_i'}\left(
\alpha_{i,j}-\beta_{i,j} \right)\left( \alpha_{i',j}-\beta_{i',j}
\right) \\ &=\frac{4\eps^2}{\ab^2}
\frac{1}{\rho_{j}^{\uniform}(1-\rho_{j}^{\uniform})}
\sum_{i=1}^{\ab/2} \left( \alpha_{i,j}-\beta_{i,j} \right)^2 \\ &=
\frac{4\eps^2}{\ab^2}\frac{\sum_{i=1}^{\ab/2} \left(
  \alpha_{i,j}-\beta_{i,j} \right)^2 }{\frac{2}{\ab}\sum_{i=1}^{\ab/2}
  \frac{ \alpha_{i,j}+\beta_{i,j} }{2} \left
  (1-\frac{2}{\ab}\sum_{i=1}^{\ab/2} \frac{ \alpha_{i,j}+\beta_{i,j}
  }{2} \right)}.
\end{align*}
To bound the expression on the right-side further, we consider the
case when $\frac{2}{\ab}\sum_{i=1}^{\ab/2} \frac{
  \alpha_{i,j}+\beta_{i,j} }{2} \leq 1/2$; the other case can be
handled similarly by symmetry. We have
\begin{align*}
\frac{\sum_{i=1}^{\ab/2} \left( \alpha_{i,j}-\beta_{i,j} \right)^2
}{\frac{2}{\ab}\sum_{i=1}^{\ab/2} \frac{ \alpha_{i,j}+\beta_{i,j} }{2}
  \left (1-\frac{2}{\ab}\sum_{i=1}^{\ab/2} \frac{
    \alpha_{i,j}+\beta_{i,j} }{2} \right)} &\leq
2\frac{\sum_{i=1}^{\ab/2} \left( \alpha_{i,j}-\beta_{i,j} \right)^2
}{\frac{2}{\ab}\sum_{i=1}^{\ab/2} \frac{ \alpha_{i,j}+\beta_{i,j}
  }{2}}\\
&\leq 2\ab \frac{\sum_{i=1}^{\ab/2} \abs{ \alpha_{i,j}-\beta_{i,j} }
  \left( \alpha_{i,j}+\beta_{i,j} \right) }{\sum_{i=1}^{\ab/2} \left(
  \alpha_{i,j}+\beta_{i,j} \right) } \\ &\leq 2\ab \max_{1\leq i\leq
  \ab/2}\abs{ \alpha_{i,j}-\beta_{i,j} } \leq 2\ab,
\end{align*}
where the previous inequality holds since $\alpha_{i,j}-\beta_{i,j}
\in [-1,1]$ for all $i,j$. Combining the foregoing bounds yields
\[
\bE{\theta}{
  \totalvardist{\mathbf{R}^{\uniform}}{\mathbf{R}^{\theta}}^2 } \leq
\frac{4\eps^2}{\ab}\cdot \ns.
\]
In particular, there exists a fixed $\theta$ for which
$\totalvardist{\mathbf{R}^{\uniform}}{\mathbf{R}^{\theta}}^2 \leq
4\eps^2n/\ab$. By the two-point method argument, the uniformity
testing protocol can only distinguish $\uniform$ and $\p_\theta$ if $
\totalvardist{\mathbf{R}^{\uniform}}{\mathbf{R}^{\theta}} =
\bigOmega{1}$, which yields $\ns=\bigOmega{\ab/\eps^2}$ as claimed.

Finally, to extend the result to public-coin protocols, note that the
observation of the referee now includes $U$ in addition to the
communication. Denote by $\mathbf{R}_U^\uniform$ and
$\mathbf{R}_U^\theta$ the distribution of the communicated bits under
$\uniform$ and $\p_\theta$, respectively. Then the total variational
distance between the distributions of the observation of the adversary
under the two distributions is given by
$\bE{U}{\totalvardist{\mathbf{R}_U^{\uniform}}{\mathbf{R}_U^{\theta}}^2
}$. Therefore, it suffices to find a uniform upper bound for the
expected value of the total variation with respect to $\theta$ for
different fixed values of the public randomness $U$. This uniform
bound can be shown to be $4\eps^2\ns/\ab$ by repeating the proof above
for every fixed $U=u$.
\end{proof}

Moving now to the case of a general $\numbits$, we can follow the
argument above to obtain an $O(\eps^2\ns 2^\numbits/\ab)$ upper bound
for the expected total variation distance. However, this only yields
an $\Omega(\ab/(2^\numbits\eps^2))$ lower bound for $\ns$, which is
off by a factor of $2^{\numbits/2}$ from the desired bound
of~\cref{theo:uniformity:shared:randomness:ub}.  The slackness in the
bound stems from the gap between the average total variation distance
and the total variation distance between the average $\p_\theta$ and
$\uniform$. Indeed, since the uniformity testing protocol can
distinguish $p_\theta$ and $\uniform$ for every $\theta$, it can also
distinguish $\bE{\theta}{p_\theta}$ and $\uniform$. Note that the
KL-divergence-based bound for total variation distance used above is
not amenable to handling the distance between a mixture of product
distribution and a fixed product distribution. Instead, we take
recourse to an argument of Pollard~\cite{Pollard:2003} which
established, in essence, the following result.
\begin{lemma}
For any two product distributions $P^\ns = P_1 \times \dots \times
P_\ns$ and $Q^\ns = Q_1 \times \dots \times Q_\ns$ on the alphabet
$\cX$,
\[
\chi^2(Q^\ns, P^\ns) = \prod_{i=1}^\ns (1+ \chi^2(Q_i, P_i)) - 1.
\]
\end{lemma}
For our application, we need to extend this to the case when the
product distribution $P^\ns$ is replaced by a mixture of product
distributions. To that end, we use the following result which is a
slight but crucial extension of this result, also described
in~\cite{Pollard:2003} (a similar observation was used
in~\cite{Paninski:08}); we include a proof for completeness.
\begin{lemma}\label{lem:mixture_chisquare}
Consider a random variable $Z$ such that for each $Z=z$ the
distribution $Q_z^\ns$ is defined as $Q_{1,z} \times \dots \times
Q_{n,z}$. Further, let $P^\ns = P_1 \times \dots \times P_\ns$ be a
fixed product distribution. Then,
\[
\chi^2(\bE{Z}{Q_Z^\ns}, P^\ns) = \bE{ZZ'}{\prod_{i=1}^\ns (1+
  {H_i(Z,Z')})} - 1,
\]
where $Z'$ is an independent copy of $Z$ and, with $\Delta_i^z$
denoting $(Q_{i,z}(X_i)-P_i(X_i))/P_i(X_i)$,
\[
H_i(z,z') = \expect{\Delta_i^z\Delta_i^{z'}},
\]
where the expectation is over $X_i$ distributed according to $P_i$.
\end{lemma}
\begin{proof}
Using the definition of $\chi^2$-distance, we have
\begin{align*}
\chi^2(\bE{Z}{Q_Z^\ns}, P^\ns) &=
\bE{P^\ns}{\left(\bE{Z}{\frac{Q_Z^\ns(X^\ns)}{P^\ns(X^\ns)}}\right)^2}
-1 \\ &=
\bE{P^\ns}{\left(\bE{Z}{\prod_{i=1}^\ns(1+\Delta_i^Z)}\right)^2} -1,
\end{align*}
where the outer expectation is for $X^\ns$ using the distribution
$P^\ns$. The product in the expression above can be expanded as
\[
\prod_{i=1}^\ns(1+\Delta_i^Z) = 1+\sum_{i\in[\ns]}\Delta_i^Z +
\sum_{i_1>i_2}\Delta_{i_1}^Z\Delta_{i_2}^Z+ \ldots,
\]
whereby we get
\begin{align*}
\chi^2(\bE{Z}{Q_Z^\ns}, P^\ns)
&=\bE{P^\ns}{\left(1+\sum_{i}\bE{Z}{\Delta_i^Z} +
  \sum_{i_1>i_2}\bE{Z}{\Delta_{i_1}^Z\Delta_{i_2}^Z}+\ldots \right)^2}
-1 \\ &= \bE{P^\ns}{\sum_{i}\bE{Z}{\Delta_i^Z} +
  \sum_{j}\bE{Z'}{\Delta_j^{Z'}} +\sum_{i,j}\bE{Z,Z'}{ \Delta_i^Z
    \Delta_j^{Z'} } \ldots }.
\end{align*}
Observe now that $\bE{P^\ns}{\Delta_{i}^z} = 0$ for every
$z$. Furthermore, $Z$ is an independent copy of $Z'$ and $\Delta_i^Z$
and $\delta_j^Z$ are independent for $i\neq j$. Therefore, the
expectation on the right-side above equals
\[
\expect{\sum_i H_i(Z,Z')+ \sum_{i_1>i_2}
  {H_{i_1}(Z,Z')}H_{i_2}(Z,Z')+\ldots} = \expect{\prod_{i=1}^\ns
  (1+H_i(Z,Z'))}-1,
\]
which completes the proof.
\end{proof}
We are now in a position to
establish~\cref{theo:uniformity:shared:randomness:lb}.
\begin{proofof}{\cref{theo:uniformity:shared:randomness:lb}}
As before, it suffices to derive a uniform upper bound for the total
variation distance between the message distributions for private-coin
protocols. In fact, it suffices to consider deterministic protocols
since for a fixed public randomness $U=u$, the protocol is
deterministic. We apply~\cref{lem:mixture_chisquare} to the
distribution of the messages for a deterministic protocol; we retain
the channel $W_j$ notation from the $\ell=1$ proof with the
understanding that it denotes a deterministic map. Note that the
messages are independent under uniform and under $\p_\theta$ (for a
fixed public randomness $U$). For our setting, $\theta$ plays the role
of $Z$ in~\cref{lem:mixture_chisquare}. Note that under uniform
observations, player $j$ sends the message $m_j\in\{0,1\}^\ell$ with
probability
\begin{align*}
    \rho_{j,m}^{\uniform}= \frac{2}{\ab}\sum_{i=1}^{\ab/2} \left(
    \frac{W_j(m|2i-1)+W_j(m|2i)}{2} \right),
\end{align*}
and under $\p_\theta$ with probability
\[
 \rho_{j,m}^{\theta} = \rho_{j,m}^{\uniform} +
 \frac{\eps}{\ab}\sum_{i=1}^{\ab/2} \theta_i\left(
 W_j(m|2i-1)-W_j(m|2i) \right).
\]
Therefore, the quantity $\Delta_{i}^\theta$ required
in~\cref{lem:mixture_chisquare} is given by
\[
\Delta_{j}^{\theta} =
\frac{\eps\sum_{i=1}^{\ab/2}\theta_i(W_j(M_j|2i)-
  W_j(M_j|2i-1))}{\sum_{i=1}^{\ab/2}(W_j(M_j|2i)+ W_j(M_j|2i-1))},
\]
where $M_j$ is the random message sent under the uniform
distribution. Consequently, we can express $H_j(\theta, \theta')$ of~\cref{lem:mixture_chisquare}, $1\leq j \leq n$, as
\begin{align*}
    H_j(\theta, \theta') &=
\frac{\eps^2}{\ab}\cdot \sum_{m\in\{0,1\}^\ell}\sum_{i_1,i_2\in [\ab/2]}\theta_{i_1}\theta'_{i_2}    \frac{ \left(
      W_j(m|2i_1-1)-W_j(m|2i_1) \right) \left(
      W_j(m|2i_2-1)-W_j(m|2i_2) \right)}{ \sum_{i=1}^{\ab/2} \left(
      W_j(m|2i-1)+W_j(m|2i) \right) }\\ 
&= \frac{\eps^2}{\ab}\cdot \theta^T H_j \theta',
\end{align*}
where $\theta^T$ denotes the transpose of the vector $\theta\in\{-1,+1\}^{k/2}$ and $H_j$ is an $[\ab/2]\times [\ab/2]$ matrix with the $(i_1, i_2)$th entry given by
\[
\sum_{m\in \{0,1\}^\ell}\frac{ \left(
      W_j(m|2i_1-1)-W_j(m|2i_1) \right) \left(
      W_j(m|2i_2-1)-W_j(m|2i_2) \right)}{ \sum_{i=1}^{\ab/2} \left(
      W_j(m|2i-1)+W_j(m|2i) \right) }.
\]
Note that the matrix $H_j$ is symmetric (in fact, it has the outer product form $A A^T$ for an $(\ab/2) \times 2^\ell$ matrix $A$). 
Therefore, we obtain from~\cref{lem:mixture_chisquare} that
\begin{align*}
\expect{\totalvardist{\bE{\theta}{\mathbf{R}^\theta}}{ \mathbf{R}^\uniform}^2}
&\leq \frac 1 4\expect{\chi^2(\bE{\theta}{\mathbf{R}^\theta}, \mathbf{R}^\uniform)}
\\
&= \frac{1}{4}\Big( \bE{\theta \theta'}{\prod_{j=1}^\ns\left(1+\frac{\eps^2}{\ab}\theta^T H_j\theta'\right)}-1 \Big)
\\
&\leq \frac{1}{4}\Big( \bE{\theta \theta'}{ e^{\frac{\ns\eps^2}{\ab}\theta^T \bar{H}\theta'}}-1 \Big)\,,
\end{align*}
where we have used $1+x \leq e^x$ and 
\[
\bar{H}= \frac 1 \ns \sum_{j=1}^\ns H_j.
\]
Thus, we need to bound the moment generating function of the random variable $\theta^T \bar{H}\theta'$. We will establish a sub-Gaussian bound using a by-now-standard bound that follows from transportation method. Specifically, we show the following:
\begin{claim}\label{claim:uniformity:lb:technical:transportation}
Consider random vectors $\theta, \theta'\in\{-1, 1\}^{\ab/2}$ with each $\theta_i$ and $\theta'_i$ distributed uniformly over $\{-1, 1\}$, independent of each other and independent for different $i$s. Then, for any symmetric matrix $H$
\[
\ln \bE{\theta \theta'}{ e^{\lambda\theta^T H\theta'}} \leq \lambda^2 \norm{H}_F^2, \quad \forall\, \lambda>0,
\]
where $\norm{\cdot}_F$ denotes the Frobenius norm. 
\end{claim}
Before we prove the claim, we use it to complete our proof. Combining this claim with our foregoing bound, we obtain
\[
\expect{\totalvardist{\bE{\theta}{\mathbf{R}^\theta}}{\mathbf{R}^\uniform}^2}
\leq\frac 1 4 \left(e^{\frac{\ns^2\eps^4}{\ab^2}\norm{\bar{H}}_F^2} -1\right)
\leq \frac 1 4 \left(e^{\frac{\ns^2\eps^4}{\ab^2}\frac 1{\ns}\sum_{j=1}^\ns\norm{H_j}_F^2} -1\right),
\]
where in the previous inequality we used the convexity of squared-norm. To complete the proof, we show now that for every $j\in [\ns]$, $\norm{H_j}^2_F\leq 2^\ell$. This is where we need to use the  assumption that the protocol is deterministic. Specifically, for every pair $m,m'\in\{0,1\}^\ell$, let
\[
    S_{m,m'} \eqdef \setOfSuchThat{ i\in [\ab/2] }{
      W_j(m|2i-1)=W_j(m'|2i)=1 }\cup \setOfSuchThat{ i\in [\ab/2] }{
      W_j(m|2i)=W_j(m'|2i-1)=1 }
\]
and
\[
    S_{m} \eqdef \setOfSuchThat{ i\in [\ab/2] }{ W_j(m|2i-1)=1 }\cup
    \setOfSuchThat{ i\in [\ab/2] }{ W_j(m|2i)=1 } =
    \bigcup_{m'\in\{0,1\}^\ell} S_{m,m'}\,.
\]
It is then immediate to see that the $S_m$'s are disjoint and that
\begin{align*} 
\norm{H_j}_F^2 &\leq \sum_{m,m'}\frac{\abs{S_{m,m'}}^2}{ \abs{S_m}\abs{S_{m'}} } \leq
\sum_{m,m'}\frac{\abs{S_{m,m'}}}{ \abs{S_m} } = \sum_m \frac{\sum_{m'}
  \abs{S_{m,m'}}}{\abs{S_m}} = 2^{\ell}\,,
\end{align*} 
whereby 
\[
\expect{\totalvardist{\bE{\theta}{\mathbf{R}^\theta}}{\mathbf{R}^\uniform}^2}
\leq \frac 1 4 \left(e^{\frac{\ns^2\eps^4 2^\ell}{\ab^2}} -1\right).
\]
The proof of the theorem can now be completed as the proof for $\ell=1$ by first noting that the same bound holds for the total variation distance even with public randomness, since we have a uniform bound for each fixed realization of public randomness, and taking recourse to  the standard two-point argument.

It only remains to establish~\cref{claim:uniformity:lb:technical:transportation}.  To that end, we use the following bound which can be obtained by combining the transportation lemma with
Marton's transportation-cost inequality (cf.~\cite[Chapter 8]{Boucheron:13}).
\begin{lemma}
Consider independent random variables $X=(X_1,\dots,X_n)$ and a
function $f$ such that for every $x,y$
\[
f(x) - f(y)\leq \sum_{i=1}^n c_i(x)\indic{x_i\neq y_i}\,.
\]
Then, setting $ v \eqdef \sum_{i=1}^n\expect{c_i^2(X)} $, we have, for
every $\lambda>0$, $ \ln \expect{e^{\lambda f(X)}} \leq \frac{\lambda^2
  v}{2}.  $
\end{lemma}
We apply this lemma to $f(Z,Z') = Z^T H Z'$, where $H$ is a symmetric
matrix and $Z,Z'$ are independent copies of $\{-1,+1\}^n$-valued
i.i.d. Rademacher vectors. In this case,
\begin{align*}
v = 2\sum_{i=1}^n \expect{\left(Z_i\sum_{j=1}^n H_{ij}\right)^2} =
2\norm{H}_F^2,
\end{align*}
which completes the proof of the claim and thereby that of~\cref{theo:uniformity:shared:randomness:lb}.
\end{proofof}

\paragraph{Acknowledgments.}
The authors would like to thank the organizers of the 2018 Information Theory and Applications Workshop (ITA), where the collaboration leading to this work started.

\clearpage
  \bibliographystyle{alpha}
  \bibliography{references}

\clearpage
\appendix
\section{From uniformity to parameterized identity testing}\label{app:identity:from:uniformity}
In this appendix, we explain how the existence of any distributed
protocol for uniformity testing implies the existence of one
for identity testing with roughly the same parameters, and
further even implies one for identity testing in the \emph{massively
parameterized} sense\footnote{Massively parameterized setting,
a terminology borrowed from property testing, refers here to the
fact that the sample complexity depends not only on a single parameter
$\ab$ but a $\ab$-ary distribution $\q$.} (``instance-optimal'' in the vocabulary of
Valiant and Valiant, who introduced it~\cite{VV:17}). These two
results will be seen as a straightforward consequence
of~\cite{Goldreich:16}, which establishes the former reduction in the
standard non-distributed setting; and of~\cite{BCG:17}, which implies
that massively parameterized identity testing reduces to
``worst-case'' identity testing. Specifically, we show the following:

\begin{proposition}\label{prop:identity:from:uniformity}
  Suppose that there exists an $\numbits$-bit protocol $\pi$ for
  testing uniformity of $\ab$-ary distributions, with number of
  players $\ns(\ab,\ell,\eps)$ and failure probability $1/3$. Then
  there exists an $\numbits$-bit protocol $\pi'$ for testing identity
  against a fixed $\ab$-ary distribution $\q$ (known to all players),
  with number of players $\ns(5\ab,\ell,\frac{16}{25}\eps)$ and
  failure probability $1/3$.
  
  Furthermore, this reduction preserves the setting of randomness
  (i.e., private-coin protocols are mapped to private-coin protocols).
\end{proposition}
\begin{proof}
  We rely on the result of Goldreich~\cite{Goldreich:16}, which
  describes a randomized mapping
  $F_{\q}\colon \distribs{[\ab]}\to\distribs{[5\ab]}$ such that
  $F_{\q}(\q)=\uniform_{[5\ab]}$ and
  $\totalvardist{F_{\q}(\p)}{\uniform_{[5\ab]}} > \frac{16}{25}\eps$
  for any $\p\in\distribs{[\ab]}$ $\eps$-far from
  $\q$.\footnote{In~\cite{Goldreich:16}, Goldreich exhibits a
  randomized mapping that converts the problem from testing identity
  over domain of size $\ab$ with proximity parameter $\eps$ to testing
  uniformity over a domain of size $\ab'\eqdef\ab/\alpha^2$ with
  proximity parameter $\eps'\eqdef (1-\alpha)^2\eps$, for every fixed
  choice of $\alpha\in(0,1)$. This mapping further preserves the
  success probability of the tester. Since the resulting uniformity
  testing problem has sample complexity
  $\bigTheta{\sqrt{\ab'}/{\eps'}^2}$, the blowup factor
  $1/(\alpha(1-\alpha)^4)$ is minimized by $\alpha = 1/5$.} In more
  detail, this mapping proceeds in two stages: the first allows one to
  assume, at essentially no cost, that the reference distribution $\q$
  is ``grained,'' i.e., such that all probabilities $\q(i)$ are a
  multiple of $1/m$ for some $m=O(\ab)$. Then, the second mapping
  transforms a given $m$-grained distribution to the uniform
  distribution on an alphabet of slightly larger cardinality. The
  resulting $F_\q$ is the composition of these two mappings.
  
  Moreover, a crucial property of $F_\q$ is that, given the knowledge
  of $\q$, a sample from $F_{\q}(\p)$ can be efficiently simulated
  from a sample from $\p$; this implies the proposition.
\end{proof}

\begin{remark}
  The result above crucially assumes that every player has explicit
  knowledge of the reference distribution $\q$ to be tested against,
  as this knowledge is necessary for them to simulate a sample from
  $F_{\q}(\p)$ given their sample from the unknown $\p$. If only the
  referee~$\referee$ is assumed to know $\q$, then the above reduction
  does not go through, although one can still rely on any testing
  scheme based on distributed simulation, as outlined
  in~\cref{sec:sampling:applications}.
\end{remark}

The previous reduction enables a distributed test for any identity
testing problem using at most, roughly, as many players as that
required for distributed uniformity testing. However, we can expect to
use fewer players for specific distributions. Indeed, in the
standard, non-distributed setting, Valiant and Valiant in~\cite{VV:17}
introduced a refined analysis termed the instance-optimal setting and
showed that the sample complexity of testing identity to $\q$ is
essentially captured by the $2/3$-quasinorm of a 
sub-function of $\q$ obtained as follows: Assuming without loss of generality $\q_1 \geq\q_2\geq \dots\q_\ab \geq 0$, 
let $t\in[\ab]$ be the largest integer that $\sum_{i=t+1}^\ab
q_i \geq \eps$, and let $\q_\eps=(\q_2,\dots,\q_t)$ (i.e., removing
the largest element and the ``tail'' of $\q$). The main result
in~\cite{VV:17} shows that the sample complexity of testing identity
to $\q$ is upper and lower bounded by 
$\max(\norm{\q_{\eps/16}}_{2/3}/\eps^2,1/\eps)$ and
$\max(\norm{\q_\eps}_{2/3}/\eps^2,1/\eps)$, respectively.

However, it is not clear if the aforementioned reduction between
identity and uniformity of Goldreich preserves this parameterization
of sample complexity for identity testing; in particular, the
$2/3$-quasinorm characterization does not seem to be amenable to the same type
of analysis as that
underlying~\cref{prop:identity:from:uniformity}. Interestingly, a
different instance-optimal characterization due to Blais, 
Canonne, and Gur~\cite{BCG:17} admits such a reduction,
enabling us to obtain the analogue
of~\cref{prop:identity:from:uniformity} for this massively
parameterized setting.\medskip

To state the result as parameterized by $\q$ (instead of $\ab$), we
will need the following definition of $\Phi(\p,\gamma)$; 
 see~\cite[Section 6]{BCG:17} for a discussion
on basic properties of $\Phi(\p,\gamma)$ and how it
relates to notions such as the sparsity of $\p$ and the
functional $\norm{\p_{\gamma}^{-\max}}$ defined in~\cite{VV:17}. For
$a\in\lp[2](\N)$ and $t\in(0,\infty)$, let
\[
    \kappa_{a}(t) \eqdef \inf_{a'+a''=a} \left( \normone{a'} +
    t\normtwo{a''} \right)
\]
and, for $\p\in\distribs{\N}$ and any $\gamma\in(0,1)$, let
\begin{equation}
    \Phi(\p,\gamma) \eqdef 2\kappa_{\p}^{-1}(1-\gamma)^2\,.
\end{equation}
It can be seen that, if $\p$ is supported on at most $\ab$
elements, $\Phi(\p,\gamma) \leq 2\ab$ for all $\gamma\in (0,1)$. We
are now in a position to state our general reduction.

\begin{proposition}\label{prop:param:identity:from:uniformity}
  Suppose that there exists an $\numbits$-bit protocol $\pi$ for testing
  uniformity of $\ab$-ary distributions, with number of players
  $\ns(\ab,\ell,\eps)$ and failure probability $1/3$. Then there
  exists an $\numbits$-bit protocol $\pi'$ for testing identity
  against a fixed distribution $\p$ (known to all players), with
  number of players
  $\bigO{ \ns(\Phi( \q, \frac{\eps}{9}),\ell, \frac{\eps}{18}) ) }$
  and failure probability $2/5$.
  
  Further, this reduction preserves the setting of randomness (i.e.,
  private-coin protocols are mapped to private-coin protocols).
\end{proposition}
\begin{proof}
  This strengthening of~\cref{prop:identity:from:uniformity} stems
   from  the algorithm for identity testing given in~\cite{BCG:17}, which at a
  high-level reduces testing identity to $\q$ to three tasks:
  (i)~computing the $(\eps/3)$-effective support\footnote{Recall
  the \emph{$\eps$-effective support} of a distribution $\q$ is the
  minimal set of elements accounting for at least $1-\eps$ probability
  mass of $\q$.}{} of $\q$, $S_{\q}(\eps)$, which can be done easily
  given explicit knowledge of $\q$; (ii)~testing that the unknown
  distribution $\p$ puts mass at
   most $\eps/2$ outside of
  $S_{\q}(\eps)$ (which only requires $O(1/\eps)$ players to be done
  with a high constant probability, say $1/30$); and (iii)~testing
  identity of $\p$ and $\q$ conditioned on $S_{\q}(\eps)$ with
  parameter $\eps/18$, which can be done using rejection sampling
  and~\cref{prop:identity:from:uniformity} with
  $\bigO{\ns(\abs{S_{\q}(\eps)},\ell,\frac{\eps}{18})}$ players and
  success probability, say $2/3-1/30$, where the additional $1/30$ error
  probability comes from rejection sampling. See~\cref{fig:distributions} for an illustration.
  
  As shown in~\cite[Section 7.2]{BCG:17}, we have 
  $\abs{S_{\q}(\eps)} \leq \Phi(\q,\frac{\eps}{9})$, and thereby the claimed
  result, since it follows that the approach above indeed yields an
  algorithm which is instance-optimal. Technically, the
     claimed  bound is obtained upon recalling that
  $\ns(\Phi(\q,\frac{\eps}{9}),\ell,\frac{\eps}{18}))= \bigOmega{1/\eps}$
  using the trivial lower bound of $\bigOmega{1/\eps}$ on uniformity
  testing, so that
  $\ns(\Phi(\q,\frac{\eps}{9}),\ell,\frac{\eps}{18}))+O(1/\eps)
  = \bigO{\ns(\Phi(\q,\frac{\eps}{9}),\ell,\frac{\eps}{18})}$.
\end{proof}

\begin{figure}[H]\centering
  \begin{tikzpicture}[x=1pt,
  y=10pt] \pgfmathsetmacro{\xmax}{300} \pgfmathsetmacro{\ymax}{10} %
  Grid 
  (\xmax,\ymax); \draw [<->] (0,\ymax) node[above] {$\q(i),\p(i)$} --
  (0,0) -- (\xmax,0) node[right] {$i$};

  \pgfmathsetseed{897273} 

  \pgfmathsetmacro{\rbuckets}{40} \pgfmathsetmacro{\buckwidth}{(\xmax/\rbuckets)} \pgfmathsetmacro{\T}{\rbuckets*0.75} \pgfmathsetmacro{\R}{\T+1}

      \node [below] at (\xmax,-0.25) {$\ab$}; \node [below] at
      (0,-0.25) {$1$}; \node [below] at ({\T*\buckwidth},-0.25)
      {$\ab_\eps$}; \node [below] at ({\T*\buckwidth/2},-0.25)
      {$S_\q(\eps)$};

    {1,...,\rbuckets}{ ({\i*\buckwidth},0) -- ({\i*\buckwidth},\ymax)
    }; \draw ({\T*\buckwidth},0) -- ({\T*\buckwidth},\ymax);

        reference, blue, one) \foreach \i in
        {1,...,\rbuckets}{ \pgfmathsetmacro{\cointoss}{random} \pgfmathsetmacro{\cointoss}{(random()-0.5)/10+(random()-1/2)*abs(\cointoss)/\cointoss}; \draw[thin,blue]
        ({(\i-1)*\buckwidth},{\ymax-1.25*sqrt(\i)}) --
        ({(\i)*\buckwidth},{\ymax-1.25*sqrt(\i)}); \draw[thin,red]
        ({(\i-1)*\buckwidth},{\ymax-1.25*sqrt(\i)+\cointoss}) --
        ({(\i)*\buckwidth},{\ymax-1.25*sqrt(\i)+\cointoss});
        }; \foreach \i in {\R,...,\rbuckets}{ \fill[thick,blue,
        opacity=0.15] ({(\i-1)*\buckwidth},0) --
        ({(\i-1)*\buckwidth},{\ymax-1.25*sqrt(\i)}) --
        ({(\i)*\buckwidth},{\ymax-1.25*sqrt(\i)}) --
        ({(\i)*\buckwidth},0) -- cycle;
          
        }; \node [] at ({(\T*\buckwidth+\xmax)/2},{(\ymax -
        1.25*sqrt((\T+\xmax/\buckwidth)/2))/2})
        {$\eps$}; \end{tikzpicture}\caption{\label{fig:distributions}The
        reference distribution $\q$ (in blue; assumed non-increasing
        without loss of generality) and the unknown distribution $\p$
        (in red). By the  reduction above, testing equality of $\p$ to
        $\q$ is tantamount to
        (i)~determining $S_\q(\eps)$, which depends only on $\q$;
        (ii)~testing identity for the conditional distributions of $\p$ and
        $\q$ given $S_\q(\eps)$, and (iii)~testing that $\p$
        assigns at most $O(\eps)$ probability to the complement of $S_\q(\eps)$.}
\end{figure}
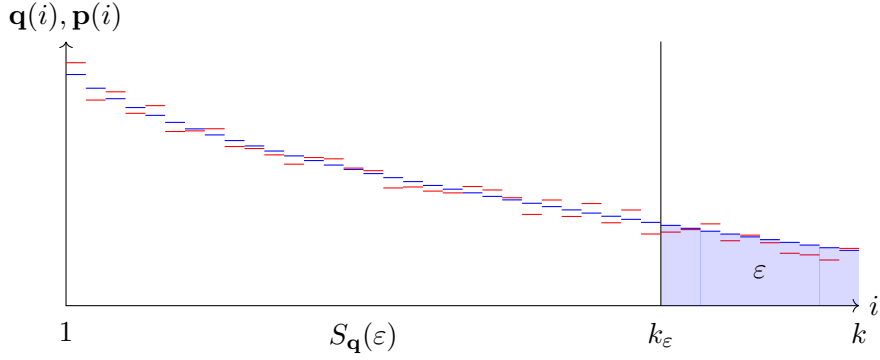
 
\section{Distributed learning lower bound (for public-randomness adaptive protocols)}\label{app:learning:lb}
\begin{theorem}\label{theo:learning:lb}
  For $1\leq \numbits\leq \log\ab$, any $\numbits$-bit public-coin (possibly adaptive) $(\ab, \eps, 1/3)$-learning protocol must have $\ns = \bigOmega{\frac{\ab^2}{2^\numbits\eps^2}}$ players.
\end{theorem}
\begin{proof}
We will show that the $\lp[1]$ minimax rate is
\[
    \inf_{(W,\delta)} \inf_{\hat{\p}} \sup_{\p\in\distribs{[\ab]}} \shortexpect_{\p} \normone{ \hat{\p} - \p } \geq C\cdot \mleft( \frac{\ab}{ \sqrt{ \ns ( 2^\numbits \wedge \ab )} }\wedge 1 \mright)
    = C\cdot \mleft(  \sqrt{ \frac{\ab}{\ns} }\vee \frac{\ab}{ \sqrt{\ns 2^\numbits} }  \wedge 1 \mright)
\]
for some absolute constant $C>0$, which implies the result. Note that
since the collocated model can simulate the distributed one, the term
$\Omega( \sqrt{\ab/\ns} )$ in the lower bound, which dominates when $2^\numbits \geq \ab$,  is an immediate consequence of the standard
lower bound in the collocated case (and holds without restriction on
the range of $\ns$). Thus, it suffices to focus on the remaining case
$2^\numbits < \ab$. 

In the argument below, we 
fix the realization shared randomness and restrict to deterministic
protocols. Our proof of lower bound uses standard argument to relate
minimax risk to probability of error in multiple hypothesis testing
problem with uniform prior on the hypotheses. Since for every
public-coin protocol we must have a deterministic protocol with same
probability of error for the multiple hypothesis testing problem,
there is no loss in restricting to deterministic protocols.

First, we establish the rate $\bigOmega{\ab/\sqrt{ \ns 2^\numbits}}$, assuming $\ns \geq \ab^2/2^\numbits$. (Recall that for $\ns \leq \ab^2/2^\numbits$, the lower bound in the RHS above is $1$.) We follow the proof of Han, \"Ozg\"ur, and Weissman~\cite[Proposition 1]{HOW:18:v1}, with the necessary modifications to adapt it to $\lp[1]$ loss (instead of squared $\lp[2]$) and remove the constraint that $\ns\geq \ab^2/2^\numbits$. (As in their proof, assume without loss of generality that $\ab$ is even.) 

To handle the dependences between the $2^\numbits$ outputted by any given player, we consider the Poissonized observation model, where we instead of $\ns$ players sending a message $Y_j$ in $\{0,1\}^\numbits$  we have 
$\ns$ players sending each a message $\tilde{Y}_j$ in $\N^{\numbits}$, where each bit of the message is a (conditionally) independent Poisson random variable: $\tilde{Y}^\ns=(\tilde{Y}_1,\dots,\tilde{Y}_\ns)\in(\N^{\numbits})^\ns$, with 
\[
    \forall j\in[\ns],\forall m\in[2^\numbits],\quad \tilde{Y}_{j,i}\mid b^{j-1} \sim \poisson{\probaCond{ Y_j=m }{ X,b^{j-1}  } }
\]
where, for $j\in[\ns]$, $b^j=(b_1,\dots,b_j)\in\{0,1\}^j$ is the (``side information'') tuple of bits with $b_j \eqdef \indic{\sum_{m=1}^{2^\numbits} \tilde{Y}_{j,m} = 1}$; and for each $j\in[\ns]$ $(\tilde{Y}_{j,1},\dots,\tilde{Y}_{j,2^\numbits})$ are independent conditioned on $b^{j-1}$. In other terms, we replace the $[2^\numbits]$-valued message of player $j$ by $2^\numbits$ different Poisson random variables, each with the right expectation (and, for technical reasons, with side information about the messages sent some other players). 
As established in Lemma 1 of~\cite{HOW:18:v1}, for distribution estimation a lower bound on the Poissonized model implies the same lower bound (up to constant factors) for our original setting.

In order to prove the lower bound, we define the family of hard instances (which will be random small perturbation of the uniform distribution $\uniform_\ab$). Letting $U$ be uniformly distributed in the hypercube $\{-1,1\}^{t}$ (where $t\eqdef \frac{\ab}{2}$), we choose $\gamma\in[0,1]$ (suitably set later in the proof) and let $\p_U\in\distribs{[\ab]}$ be defined by its probability mass function
\[
    \p_U = \frac{1}{\ab}\left( 1+\gamma U_1,\dots,1+\gamma U_{t}, 1-\gamma U_1,\dots, , 1-\gamma U_t\right)\,.
\]
This defines a class $\class\subseteq\distribs{[\ab]}$ of $2^t$ distributions. Since clearly the $\lp[1]$ minimax risk over \emph{all} $\ab$-ary distributions is no less than that over $\class$, it suffices to lower bound the later.
 We will rely on the following lemma to first bound the mutual information between the tuple of Poissonized messages $\tilde{Y}^\ns$ and the unknown parameter $U$ to estimate:
\begin{lemma}[{\cite[Lemma 3]{HOW:18:v1}}]\label{how18:lemma3}
  The following upper bound holds:
  \[
      \mutualinfo{U}{\tilde{Y}^\ns} \leq 2 \sum_{j=1}^\ns \sum_{m=1}^{2^\numbits} \shortexpect_{U,U'}\left[ \frac{(\proba_{\p_U}[ Y_j=m \mid X_j,b^{j-1}  ] - \proba_{\p_{U'}}[ Y_j=m \mid X_j,b^{j-1}  ] )^2}{\shortexpect_U\proba_{\p_U}[ Y_j=m \mid X_j,b^{j-1}  ] } \right]
  \]
  where $U'$ is an independent copy of $U$.
\end{lemma}
To handle the right-hand-side of the above bound, observe that any randomized strategy $W\colon [\ab]\to\{0,1\}$ can be identified with a vector $w\in[0,1]^{\ab}$. For every such $w$, we have
\begin{align}
  \shortexpect_{U,U'} \frac{(\shortexpect_{\p_U} W(Y\mid X) - \shortexpect_{\p_U'} W(Y\mid X))^2}{\shortexpect_U\shortexpect_{\p_U} W(Y\mid X)}
  &= \ab \frac{w^T \shortexpect_{U,U'}[(\p_U-\p_{U'})(\p_U-\p_{U'})^T] w}{w^T \mathbf{1}} \notag\\
  &\leq \ab\frac{4\gamma^2}{\ab^2}\cdot\frac{w^T w}{w^T \mathbf{1}} \leq \frac{4\gamma^2}{\ab} \label{eq:lb:how:adapted}
\end{align}
the last step since $\norminf{w}\leq 1$. We will use this later on, after relating this mutual information $\mutualinfo{U}{\tilde{Y}^\ns}$ to the quantity we are trying to analyze, the $\lp[1]$ minimax risk over our class $\class$~--~which we do next. It is not hard to show, via a standard ``Assouad's Lemma''-type argument that this $\lp[1]$ minimax risk can be lower bounded as
\begin{equation}\label{eq:lb:how:adapted:minimax:lb}
    \inf_{\hat{\p}} \sup_{\p\in\class} \shortexpect_{\p} \normone{ \hat{\p} - \p } \geq c\cdot \gamma \inf_{\hat{U}} \probaOf{ \dist{\hat{U}}{U} \geq t/5 }
\end{equation}
where $\dist{\cdot}{\cdot}$ is the unnormalized Hamming distance and $U$ is a uniform random vector in $\{-1,1\}^t$ and $c>0$ is an absolute constant. (This is another part where we depart from the argument of~Han, \"Ozg\"ur, and Weissman, concerned with the squared $\lp[2]$ loss.) Invoking~\cref{how18:lemma3}, along with~\eqref{eq:lb:how:adapted} and the same distance-based Fano's inequality as in~\cite[Lemma 2]{HOW:18:v1}, we can conclude that
\[
  \inf_{\hat{U}} \probaOf{ \dist{\hat{U}}{U} \geq t/5 } \geq 1- \frac{ \mutualinfo{U}{Y^\ns}+\ln 2 }{ t/8 }
  \geq 1- \frac{ 2\cdot \ns 2^\numbits \cdot \frac{4\gamma^2}{\ab} +\ln 2 }{ t/8 }
  = 1 - 16\frac{ 8\gamma^2 \ns 2^\numbits + \ab\ln 2 }{\ab^2}\,.
\]
The RHS will be at least say $1/2$, for large enough $\ab$, by setting $\gamma^2 \eqdef c'\cdot \frac{\ab^2}{\ns 2^\numbits}$ for a constant $c'>0$ sufficiently small (but independent of $\ab,\ns,\numbits$). For this choice of $\gamma$,~\eqref{eq:lb:how:adapted:minimax:lb} becomes
\[
    \inf_{\hat{\p}} \sup_{\p\in\class} \shortexpect_{\p} \normone{ \hat{\p} - \p } \geq \frac{c}{2}\gamma = C\cdot\frac{\ab}{\sqrt{\ns 2^\numbits}}
\]
(where $C \eqdef \frac{c\cdot \sqrt{c'}}{2} >0$), concluding the proof. (Note that the constraint $\ns \geq \ab/2^\numbits$ was used in the setting of $\gamma^2$, to ensure that $\gamma\in[0,1]$.)

Finally, we are left with the case $\ns \leq \ab^2/2^\numbits$, where we must show that the rate is $\bigOmega{1}$. We can prove it by reducing it to the previous case: namely, divide the domain $[\ab]$ into $\ab'\eqdef\sqrt{2^\numbits \ns}<\ab$ disjoint intervals of equal size (assuming for simplicity, and with little loss of generality, that $\ab'$ divides $\ab$). Apply now the previous construction to the induced domain over $\ab'$ elements, setting the distribution $\p_U$ to be uniform on each of the $\ab'$ intervals. This leads to the setting of $\gamma^2 = \frac{{\ab'}^2}{\ns 2^\numbits}\in[0,1]$, and a lower bound on the risk of $\Omega(\gamma)=\Omega(1)$.
\end{proof}

\section{Proof of~\cref{theo:uniformity:z:concentration:anticoncentration:general}}\label{app:smooth}
In this appendix, we
prove~\cref{theo:uniformity:z:concentration:anticoncentration:general},
stating that taking a random balanced partition of the domain in
$L\geq 2$ parts preserves the $\lp[2]$ distance between distributions
with constant probability. Note that, as mentioned
in~\cref{sec:uniformity:ub:smooth}, the special case of $L=2$ was
proven in~\cite{ACFT:18}. In fact, the proof for general $L$ is similar to the proof in \cite{ACFT:18}, but requires some additional work. We provide a self-contained proof here for easy reference.

We begin by recall the Paley--Zigmund inequality, a key tool we shall rely upon.
\begin{theorem}[Paley--Zygmund]\label{theo:paley:zygmund}
    Suppose $U$ is a non-negative random variable with finite variance. Then, for every $\theta\in[0,1]$, 
    \begin{equation*}
        \probaOf{ U > \theta\expect{U} } \geq (1-\theta)^2\frac{\expect{U}^2}{\expect{U^2}}\,.
    \end{equation*}
\end{theorem}

We will prove a more general version of~\cref{theo:uniformity:z:concentration:anticoncentration:general}, showing that the $\lp[2]$ distance to any fixed distribution $\q\in\distribs{[\ab]}$ is preserved with a constant probability.\footnote{For this application, one should read the theorem statement with $\delta \eqdef \p-\q$.}
Let random variables $X_1,\dots, X_\ab$ be as in~\cref{theo:uniformity:z:concentration:anticoncentration:general}; in particular, each $X_i$ is distributed uniformly on $[L]$ and for every $r\in[L]$, $\sum_{i=1}^\ab \indic{X_i=r} = \frac{\ab}{L}$.

\begin{theorem}\label{theo:uniformity:anticoncentration:general}
Suppose $2\leq L< \ab$ is an integer dividing $\ab$, and fix $\delta\in\R^\ab$ such that $\sum_{i\in[\ab]} \delta_i = 0$. For random variables $X_1, ..., X_\ab$ above, let $Z=(Z_1,\dots,Z_L)\in \R^L$ with
    \begin{equation*}
        Z_r \eqdef \sum_{i=1}^\ab \delta_i \indic{X_i=r},\qquad r\in[L]\,.
    \end{equation*}
Then,   there exists a constant $c>0$ such that
  \begin{equation*}
      \probaOf{ \normtwo{Z} > \frac 12 \cdot\normtwo{\delta}} \geq c.
  \end{equation*}

\end{theorem}

\begin{proofof}{\cref{theo:uniformity:anticoncentration:general}}

As in~\cite[Theorem 14]{ACFT:18}, the gist of the proof is to consider a suitable non-negative random variable (namely, $\normtwo{Z}^2$) and bound its
expectation and second moment in order to apply the Paley--Zygmund
inequality to argue about anticoncentration around the mean. The
difficulty, however, lies in the fact that bounding the moments of
$\normtwo{Z}$ involves handling the products of correlated $L$-valued
random variables $X_i$'s, which is technical even for the case $L=2$
considered in~\cite{ACFT:18}. For ease of presentation, we have divided the proof into smaller results. 

\begin{lemma}[Each part has the right expectation]\label{lemma:expectation:z:general:vanilla}
  For every $r\in[L]$, 
  \[
      \expect{Z_r} = 0\,.
  \]
\end{lemma}
\begin{proof}
By linearity of expectation,
\[
    \expect{Z_r}  = \sum_{i=1}^\ab \delta_i  \expect{ \indic{X_i=r} } = \frac{1}{L} \sum_{i=1}^\ab \delta_i = 0
.
\] 
\end{proof}

\begin{lemma}[The {$\lp[2]^2$} distance to uniform of the flattening has the right expectation]\label{lemma:variance:z:general:vanilla}
For every $r\in[L]$,
  \[
      \var Z_r = \expect{Z_r^2} = \frac{1}{L}\normtwo{\delta}^2\left( 1 -\frac{1}{L} + \frac{L-1}{L(\ab-1)} \right) \geq \frac{1}{2L}\normtwo{\delta}^2\,.
  \]
  In particular, the expected squared $\lp[2]$ norm of $Z$ is
  \begin{equation*}
      \expect{\normtwo{Z}^2 } = \expect{\sum_{r=1}^L Z_r^2 } \geq \frac{1}{2}\normtwo{\delta}^2\,.
  \end{equation*}
\end{lemma}
\begin{proof}
For a fixed $r\in[L]$, using the definition of $Z$, the fact that $\sum_{i=1}^\ab  \indic{X_i=r}  = \frac{\ab}{L}$, and~\cref{lemma:expectation:z:general:vanilla}, we get that
  \begin{align*}
      \var[ Z_r ] 
      &= \expect{ Z_r^2 } 
      = \expect{ \left(\sum_{i=1}^\ab \delta_i \indic{X_i=r} \right)^2 }
      = \sum_{1\leq i,j\leq \ab} \delta_i\delta_j \expect{\indic{X_i=r}\indic{X_j=r}} \\
      &= \sum_{i=1}^\ab \delta_i^2 \expect{ \indic{X_i=r} } + 2\sum_{1\leq i<j \leq \ab} \delta_i\delta_j \expect{ \indic{X_i=r}\indic{X_j=r} }\,.
  \end{align*}
  Since the $X_i$'s~--~while not independent~--~are identically distributed, it is enough by symmetry to compute $\expect{ \indic{X_\ab=r} }$ and $\expect{ \indic{X_{\ab-1}=r}\indic{X_\ab=r} }$. The former is $1/L$; for the latter, note that
  \begin{align}
      \expect{ \indic{X_{\ab-1}=r}\indic{X_\ab=r} }
      &= \expect{ \expectCond{ \indic{X_{\ab-1}=r}\indic{X_\ab=r} }{ \indic{X_\ab=r} } } 
      = \frac{1}{L}\probaCond{ X_{\ab-1}=r}{X_\ab = r} 
\nonumber
\\
      &= \frac{1}{L}\probaCond{ X_{\ab-1} = r}{\sum_{i=1}^{\ab-1} \indic{X_i=r} = \frac{\ab}{L}-1}
      = \frac{1}{L^2}\cdot \frac{\ab-L}{\ab-1},
\label{e:symmetry_for_conditional_prob}
  \end{align}
where the final identity uses symmetry once again, along with the observation that 
\[
\sum_{i=1}^{\ab-1} \expectCond{ \indic{X_i=r} }{\sum_{j=1}^{\ab-1} \indic{X_j=r} = \frac{\ab}{L}-1} = \frac{\ab}{L}-1.
\]
 Putting it together, we get the result as follows:
  \begin{align*}
      \var[ Z_r ] 
      &= \frac{1}{L}\sum_{i=1}^\ab \delta_i^2 + \frac{1}{L^2}\cdot \frac{\ab-L}{\ab-1}\cdot 2\sum_{1\leq i<j \leq \ab} \delta_i\delta_j 
      = \frac{1}{L}\normtwo{\delta}^2 - \frac{1}{L^2} \mleft( 1 - \frac{L-1}{\ab-1}\mright)\normtwo{\delta}^2 \\
      &= \frac{1}{L}\normtwo{\delta}^2\left( 1 -\frac{1}{L} + \frac{L-1}{L(\ab-1)} \right). 
  \end{align*}
\end{proof}

\begin{lemma}[The {$\lp[2]^2$} distance to uniform of the flattening has the required second moment]\label{lemma:fourth:moment:z:general:vanilla}
  There exists an absolute constant $C>0$ such that
  \begin{equation*}
      \expect{\normtwo{Z}^4} \leq C\normtwo{\delta}^4\,.
  \end{equation*}
\end{lemma}
\begin{proofof}{\cref{lemma:fourth:moment:z:general:vanilla}}
Expanding the square, we have
\begin{equation}\label{eq:square:expanding}
    \expect{\normtwo{Z}^4} = \expect{\mleft(\sum_{r=1}^L Z_r^2\mright)^2}
    = \sum_{r=1}^L \expect{Z_r^4} + 2\sum_{r<r'} \expect{Z_r^2Z_{r'}^2}
\end{equation}
We will bound both terms separately. For the first term, we note that
using \cite[Equation(21)]{ACFT:18} with $\indic{X_i=r}$ in the role of
$X_i$ there, each term $\expect{Z_r^4}$ is bounded above by
$19\normtwo{\delta}^4/L$ whereby 
\begin{align}
\sum_{r=1}^L \expect{Z_r^4} \leq 19 \normtwo{\delta}^4.
\label{e:first_term}
\end{align}
However, we need additional work
to handle the second term comprising roughly $L^2$ summands. In
particular, to complete the proof we show that
each summand in the second term is less than a constant factor times
$\normtwo{\delta}^4/L^2$.

\begin{claim}\label{claim:fourth:moment:z:general:vanilla:second:term}
  There exists an absolute constant $C'>0$ such that
  \[
      \sum_{r<r'} \expect{Z_r^2Z_{r'}^2} \leq C'\normtwo{\delta}^4\,.
  \]
\end{claim}
\begin{proof}
Fix any $r\neq r'$. As before, we expand
\begin{align*}
    \expect{Z_r^2Z_{r'}^2}
    &= \expect{ \left(\sum_{i=1}^\ab \delta_i \indic{X_i=r} \right)^2\left(\sum_{i=1}^\ab \delta_i \indic{X_i=r'} \right)^2 } \\
    &= \sum_{1\leq a,b,c,d \leq \ab} \delta_a\delta_b\delta_c\delta_d\expect{ \indic{X_a=r}\indic{X_b=r}\indic{X_c=r'}\indic{X_d=r'} }\,.
\end{align*}
Using symmetry once again, note that the term
$\expect{ \tilde{X}_a\tilde{X}_b\tilde{X}_c\tilde{X}_d }$ depends only
on the number of distinct elements in the multiset $\{a,b,c,d\}$,
namely the cardinality $\abs{\{a,b,c,d\}}$. The key observation here
is that if $\{a,b\}\cap\{c,d\} \neq \emptyset$, then
$\indic{X_a=r}\indic{X_b=r}\indic{X_c=r'}\indic{X_d=r'} = 0$. This
will be crucial as it implies that the expected value can only be
non-zero if $\abs{\{a,b,c,d\}}\geq 2$, yielding a $1/L^2$ dependence
for the leading term in place of $1/L$. 
 \begin{align}
    \expect{Z_r^2Z_{r'}^2}
    &= \sum_{\abs{\{a,b,c,d\}}=2} \delta_a^2\delta_b^2\expect{ \indic{X_a=r}\indic{X_b=r'}} 
\notag
\\
    &\qquad+\sum_{\abs{\{a,b,c,d\}}=3} \delta_a^2\delta_b\delta_c\expect{ \indic{X_a=r}\indic{X_b=r'}\indic{X_c=r'}} \notag\\
    &\qquad+\sum_{\abs{\{a,b,c,d\}}=3} \delta_a\delta_b\delta_c^2\expect{ \indic{X_a=r}\indic{X_b=r}\indic{X_c=r'}} \notag\\
    &\qquad+\sum_{\abs{\{a,b,c,d\}}=4} \delta_a\delta_b\delta_c\delta_d\expect{ \indic{X_a=r}\indic{X_b=r}\indic{X_c=r'}\indic{X_d=r'}}\,. 
\label{eq:cross:terms}
\end{align}
The first term, which we will show dominates, is bounded as
\[
\sum_{\abs{\{a,b,c,d\}}=2} \delta_a^2\delta_b^2\expect{ \indic{X_a=r}\indic{X_b=r'}}
= \expect{ \indic{X_{\ab-1}=r}\indic{X_\ab=r'} } \normtwo{\delta}^4 \leq \frac{2}{L^2}\normtwo{\delta}^4
\]
where the inequality uses
\[
\expect{ \indic{X_{\ab-1}=r}\indic{X_\ab=r'} }
= \frac{1}{L^2}\cdot \frac{\ab}{\ab-1} \leq \frac{2}{L^2},
\]
which in turn is obtained in the manner of \eqref{e:symmetry_for_conditional_prob}.

For the second and the third terms, noting that 
\[
\expect{ \indic{X_a=r}\indic{X_b=r'}\indic{X_c=r'} }
= \abs{\delta_a^2\delta_b\delta_c}\cdot \frac{1}{L^3} \frac{\ab(\ab-L)}{(\ab-1)(\ab-2)}, 
\]
and that 
\[
 \sum_{\abs{\{a,b,c,d\}}=3} \delta_a^2\delta_b\delta_c = \sum_{1\leq a,b,c\leq \ab} \delta_a^2\delta_b\delta_c - \sum_{a\neq b}\delta_a^2\delta_b^2 - 2\sum_{a\neq b}\delta_a^3\delta_b
\]
with $\sum_{1\leq a,b,c\leq \ab} \delta_a^2\delta_b\delta_c = \left(\sum_{a=1}^\ab \delta_a^2\right)\left(\sum_{a=1}^\ab \delta_a\right)^2 = 0$, $\sum_{a\neq b}\delta_a^2\delta_b^2\leq \sum_{1\leq a,b\leq \ab} \delta_a^2\delta_b^2 = \normtwo{\delta}^4$, and $\sum_{a\neq b}\delta_a^3\abs{\delta_b}\leq \sum_{1\leq a,b\leq \ab}\delta_a^3\abs{\delta_b} \leq \norminf{\delta}\norm{\delta}_3^3\leq \normtwo{\delta}^4$, we get
\[
      -\frac{6}{L^3}\normtwo{\delta}^4 \leq  \sum_{\abs{\{a,b,c,d\}}=3} \delta_a^2\delta_b\delta_c \expect{ \indic{X_a=r}\indic{X_b=r'}\indic{X_c=r'}}  \leq \frac{6}{L^3}\normtwo{\delta}^4\,.
\]
Finally, as $\expect{ \indic{X_a=r}\indic{X_b=r}\indic{X_c=r'}\indic{X_d=r'}} = \frac{1}{L^4}\frac{\ab^2(\ab-L)^2}{(\ab-1)(\ab-2)(\ab-3)(\ab-4)} \leq \frac{10}{L^4}$, similar manipulations yield
\[
    -\frac{\alpha}{L^4}\normtwo{\delta}^4
    \leq \sum_{\abs{\{a,b,c,d\}}=4} \delta_a\delta_b\delta_c\delta_d\expect{ \indic{X_a=r}\indic{X_b=r}\indic{X_c=r'}\indic{X_d=r'}}
    \leq \frac{\alpha}{L^4}\normtwo{\delta}^4
\]
for some absolute constant $\alpha>0$. Gathering all this in~\eqref{eq:cross:terms}, we get that there exists some absolute constant $C'>0$ such that
\[
    \sum_{r<r'} \expect{Z_r^2Z_{r'}^2}
    \leq C'\sum_{r<r'} \frac{1}{L^2}\normtwo{\delta}^4
    \leq \frac{C'}{2} \normtwo{\delta}^4\,.
\]
\end{proof}
The lemma follows by combining the previous claim with \eqref{e:first_term}.
\end{proofof}
We are now ready to establish~\cref{theo:uniformity:z:concentration:anticoncentration:general}. By~\cref{lemma:variance:z:general:vanilla,lemma:variance:z:general:vanilla,lemma:fourth:moment:z:general:vanilla}, we have $\expect{\normtwo{Z}^2} \geq \frac{1}{2}\normtwo{\delta}^2$ and $\expect{\normtwo{Z}^4} \leq C\normtwo{\delta}^4$, for some absolute constant $C>0$. Therefore, by the Payley--Zygmund inequality (\cref{theo:paley:zygmund}) applied to $\normtwo{Z}^2$ for $\theta=1/2$,
\begin{equation*}
  \probaOf{ \normtwo{Z}^2 > \frac{1}{4}\normtwo{\delta}^2 } 
  \geq   \probaOf{ \normtwo{Z}^2 > \frac{1}{2}\expect{\normtwo{Z}^2} }
  \geq \frac{1}{4}\frac{\expect{\normtwo{Z}^2}^2}{\expect{\normtwo{Z}^4}} \geq \frac{1}{16C}\,.
\end{equation*}
This concludes the proof.
\end{proofof}

\end{document}